\documentclass[journal]{IEEEtran}
\usepackage{booktabs}
\usepackage[utf8]{inputenc}
\usepackage{cite}
\usepackage[T1]{fontenc}
\usepackage[ruled, vlined, linesnumbered]{algorithm2e}
\usepackage{amsmath}
\usepackage{amsthm}
\usepackage{txfonts}
\usepackage{xcolor}
\usepackage{comment}
\usepackage{subfigure}
\usepackage{tikz}
\usepackage{color}
\usepackage{paralist} 
\usetikzlibrary{arrows,automata}
\usepackage{amsfonts}
\usepackage{datetime}
\usepackage{soul}
\usepackage{scalefnt}
\usepackage{multirow}
\usepackage{paralist}
\usepackage[protrusion=true,expansion=true]{microtype}
\usepackage{url}
\usepackage{subfigure}
\usepackage{mathrsfs}
\usepackage{amsfonts}
\usepackage{cleveref}
\usepackage{capt-of}
\usepackage{lipsum}
\usepackage{multicol}
\usepackage{multirow}
\usepackage{lscape}
\usepackage{rotating}
\usepackage{layout}
\usepackage[inline]{enumitem}

\newtheorem{theorem}{Theorem}
\newtheorem{definition}{Definition}
\newtheorem{lemma}{Lemma}

\def\TAB{Table}
\def\FIG{Fig.}

\def\SEC{Section}
\def\TAB{Table}
\def\EQ{Equation}

\def\DEF{Definition}

\newcommand{\code}{\texttt}
\newcommand{\SYM}{\textsf}

\DeclareMathOperator*{\prd}{\mathrm{Prd}}
\DeclareMathOperator*{\cns}{\mathrm{Cns}}
\DeclareMathOperator*{\pred}{\mathrm{prec}}

\def\SADF{MADF}
\def\FSADF{FSM-SADF}

\DeclareMathOperator{\mc}{\textit{MC}}

\begin{document}

\title{Modeling, Analysis, and Hard Real-time Scheduling of Adaptive Streaming Applications}

\author{Jiali~Teddy~Zhai,~Sobhan~Niknam,~and~Todor~Stefanov,~\IEEEmembership{Member,~IEEE}	\thanks{This article was presented in the International Conference on Embedded Software 2018 and appears as part of the ESWEEK-TCAD special issue.}
	\thanks{The authors contributed to the paper equally. The authors are with the Leiden Institute of Advanced Computer Science, Leiden University, Leiden, The Netherlands, Email: teddyzhai@gmail.com, \{s.niknam,t.p.stefanov\}@liacs.leidenuniv.nl. This research is supported by the Dutch Technology Foundation STW under the Robust Cyber Physical Systems program (Project 12695).}
}

\markboth{IEEE TRANSACTIONS ON COMPUTER-AIDED DESIGN OF INTEGRATED CIRCUITS AND SYSTEMS}{Shell \MakeLowercase{\textit{et al.}}: Bare Demo of IEEEtran.cls for IEEE Journals}

\maketitle
\begin{abstract}
	In real-time systems, the application's behavior has to be predictable at compile-time to guarantee timing constraints. However, modern streaming applications which exhibit adaptive behavior  due to mode switching at run-time, may degrade system predictability due to unknown behavior of the application during mode transitions. Therefore, proper temporal analysis during mode transitions is imperative to preserve system predictability.
	To this end, in this paper, we initially introduce Mode Aware Data Flow (MADF) which is our new predictable Model of Computation (MoC) to efficiently capture the behavior of adaptive streaming applications. Then, as an important part of
	the operational semantics of MADF, we propose the Maximum-Overlap Offset (MOO) which is our novel protocol for mode transitions. The main advantage of this transition protocol is that, in contrast to self-timed transition protocols, it 	avoids timing interference between modes upon mode transitions. As a result, any mode transition can be analyzed independently from the mode transitions that occurred in the past. Based on this transition protocol, we propose a hard real-time analysis as well to guarantee timing constraints by avoiding processor overloading during mode transitions. Therefore, using this protocol, we can derive a lower bound and an upper bound on the earliest starting time of the tasks in the new mode during mode transitions in such a way that hard real-time constraints are respected.
\end{abstract}

\IEEEpeerreviewmaketitle

\section{Introduction}\label{sec:Introduction}

\IEEEPARstart{T}{o} handle the ever-increasing computational demands and meet hard real-time constraints in streaming applications, where the huge amount of streaming data should be processed in a short time interval, embedded systems have relied on Multi-Processor System-on-Chip (MPSoC) platforms to benefit from parallel processing. To efficiently exploit the computational capacity of MPSoCs, however, streaming applications must be expressed primarily in a parallel fashion. The common practice for expressing the parallelism in an application is to use parallel Models of Computation (MoCs) \cite{Gerstlauer_2009}. Within a parallel MoC, a streaming application is modeled as a directed graph, where graph nodes represent actors (i.e., tasks) and graph edges represent data dependencies. Actors are executed concurrently and communicate data explicitly via FIFOs. For example, Synchronous Data Flow (SDF) \cite{Lee_1987} and Cyclo-Static Data Flow (CSDF) \cite{Bilsen1996} are two popular parallel MoCs because of their compile-time analyzability. Due to the static nature of SDF and CSDF MoCs, the actors are restricted to produce and consume data with fixed rates per firing or, in case of CSDF, with fixed periodic patterns. 

Nowadays, many modern streaming applications, in the domain of multimedia, image, and signal processing, increasingly show adaptive behavior at run-time. For example, a computer vision system processes different parts of an image continuously to obtain information from several regions of interest depending on the actions taken by the external environment. This adaptive behavior, however, cannot be effectively expressed with an SDF or CSDF model due to their limited expressiveness.
As a result, more expressive models, e.g., Scenario-Aware Data Flow (SADF) \cite{Theelen2006_memocode}, Finite State Machine (FSM)-based Scenario-Aware Data Flow (FSM-SADF) \cite{Geilen2010_codes10}, Variable-rate Phased Data Flow (VPDF)\cite{Wiggers2011_tecs}, and Mode-controlled Data Flow (MCDF)\cite{Moreira2012}, have been proposed and deployed as extensions of the (C)SDF model. These MoCs are able to capture the behavior of an adaptive streaming application as a collection of different static behaviors, called scenarios or modes, which are individually predictable in performance and resource usage at compile-time.

Moreover, to guarantee tight timing constraints in modern streaming applications with adaptive behavior nature, proper temporal analysis for application execution during mode transitions, when the application's behavior is switching from one mode to another mode, is imperative at compile-time. However, such analysis can be difficult due to the fact that different actors in different modes are concurrently executing during mode transitions. This difficulty comes directly from the protocol adopted for the mode transitions. In the existing adaptive MoCs, like MCDF \cite{Moreira2012} and FSM-SADF \cite{Geilen2010_codes10}, a protocol, referred as self-timed transition protocol, has been adopted which specifies that actors are scheduled as soon as possible not only in each mode individually, but also during mode transitions. This protocol, however, introduces interference of one mode execution with another one, as explained in \SEC~\ref{sec:st_trans}. As a consequence, the temporal analysis of a mode transition is tightly dependent on the mode transitions that occurred in the past. Another consequence of the incurred interference between modes is the high time complexity of analyzing mode transitions, as the mode transitions cannot be analyzed independently, e.g., see the state-space exploration approach proposed in \cite{Geilen2010_codes10}.

Therefore, to overcome the aforementioned interference issue and consequent problems caused by the self-timed transition protocol, in this paper, we propose a new MoC called Mode Aware Data Flow (MADF) to model adaptive streaming applications, that is armed by a novel transition protocol called Maximum-Overlap Offset. This transition protocol enables an independent analysis for mode transitions. The specific novel contributions of this paper are the following:
\begin{itemize}
	\item  We propose a new MoC, Mode-Aware Data Flow (MADF), that has the advantages of SADF \cite{Theelen2006_memocode} and VPDF \cite{Wiggers2011_tecs}. Inspired by SADF, we characterize the behavior of adaptive streaming applications with individual modes 	and transitions between them. Similar to VPDF, the length of production/consumption sequences for an actor varies from one mode to another. The length is only fixed when the mode is known. Then, based on the clear distinction between modes and transitions, we define analyzable operational semantics for MADF;
	\item As an important part of the operational semantics of MADF, we propose the Maximum-Overlap Offset (MOO) which is our novel protocol for mode transitions. The main advantage of this transition protocol is that, in contrast to the self-timed transition protocol, adopted in \cite{Geilen2010_codes10,Moreira2012}, it avoids timing interference between modes upon mode transitions. 
	As a result, this transition protocol enables an independent analysis for mode transitions. This means, the analysis of any mode transition is independent from the mode transitions that occurred in the past. This independent  analysis significantly reduces the complexity of the analysis as the complexity merely depends on the number of allowed transitions. This is crucial for applications with a large number of modes and possible transitions;
	\item Based on the novel MOO transition protocol, we propose a hard real-time analysis approach to guarantee the timing constraints by avoiding processor overloading, i.e., avoiding that the total utilization of allocated tasks on a processor exceeds its capacity, during mode transitions. Our analysis is much simpler and faster than the computationally intensive state-of-the-art timing analysis approaches such as \cite{Geilen2010_codes10}.
\end{itemize}

The remainder of this paper is organized as follows: Section \ref{sec:Related Work} gives an overview of the related work. Section \ref{background} introduces the background needed for understanding the contributions of this paper. Our novel adaptive MoC and transition protocol are then introduced in Section \ref{sec:dyn_moc_def}. Based on the novel transition protocol, in Section \ref{sec:hrt_sacsdf}, we present our hard real-time analysis approach to guarantee the timing constraints during mode transitions. In Section \ref{sec:Case Study}, two case studies are presented to illustrate the practical applicability of our proposed MADF mode, transition protocol, and real-time analysis. Finally, Section \ref{sec:Conclusion} ends the paper with conclusions.
 \section{Related Work} \label{sec:Related Work}

To model the adaptive behavior of modern streaming applications while having certain degree of compile-time analyzability, different MoCs such as Scenario-Aware Data Flow (SADF) \cite{Theelen2006_memocode},
 Finite State Machine (FSM)-based Scenario-Aware Data Flow (FSM-SADF) \cite{Geilen2010_codes10}, Variable-rate Phased Data Flow (VPDF) \cite{Wiggers2011_tecs}, Mode-controlled Data Flow (MCDF) \cite{Moreira2012}, and Parameterized SDF (PSDF) \cite{Bhattacharya2001} have been already proposed in the literature.

In SADF \cite{Theelen2006_memocode} and FSM-SADF \cite{Geilen2010_codes10}, \textit{detector} actors are introduced to parameterize the SDF model. All valid scenarios and their possible order of occurrence, which is shown either by using a Markov chain \cite{Theelen2006_memocode} or finite state machine \cite{Geilen2010_codes10}, must be predefined at compile-time. Each scenario consists of a set of valid parameter combination that determines a scenario of SADF. This guarantees the consistency of SADF in individual scenarios, therefore, no run-time consistency check is required. In a scenario, the SADF model behaves the same way as the SDF model. Therefore, an SADF graph can be seen as a set of SDF graphs. In the initial FSM-SADF definition, all the production and consumption rates of the data-flow edges are constant within a graph iteration of a scenario. 

For the FSM-SADF MoC \cite{Geilen2010_codes10}, the authors proposed an approach to compute worst-case performance among all mode transitions, assuming the self-timed transition protocol. Although it is an exact analysis, the approach has inherently exponential time complexity. Moreover, this approach leads to timing interference between modes upon mode transitions. In contrast, our approach does not introduce interference between modes due to the novel MOO transition protocol proposed in \Cref{sec:moo_trans}. The timing behavior of individual modes and during mode transitions can be analyzed independently. In addition, our approach considers allocation of actors on processors, which by itself is a harder problem than the one addressed in \cite{Geilen2010_codes10}.

In \cite{Geilen2011_tecs}, the author proposes to use a linear model to capture worst-case transition delay and period during scenario transitions of FSM-SADF.
Our transition protocol is conceptually similar to the linear model. However, we obtain the linear model in a different way, specifically simplified for the adopted hard real-time scheduling framework. For instance, finding a reference schedule is not necessary in our case, but being crucial in the tightness of the analysis proposed in \cite{Geilen2011_tecs}. Moreover, our approach solves the problem of changing the application graph structure during mode transitions, which was not studied in \cite{Geilen2011_tecs}.

For VPDF~\cite{Wiggers2011_tecs}, the analysis has been limited to computing buffer sizes under throughput constraints so far.
The execution of a VPDF graph on MPSoC platforms under hard real-time constraints has not been studied. In particular, the allocation of actors and how to switch from one mode to another one are not discussed.
Moreover, delay due to mode transitions has not been investigated.
Our approach, on the other hand, takes these important factors into account. Therefore, our analysis results are directly reflected in a real implementation.

MCDF \cite{Moreira2012} is another adaptive MoC which properties can be partly analyzed at compile-time.
The MCDF MoC primarily focuses on Software-Defined Radio applications, where different sub-graphs need to be active in different modes.
This is achieved by using \textit{switch} and \textit{select} actors.
The author implicitly assumes self-timed scheduling during mode transitions.
Based on this assumption, a worst-case timing analysis is developed.
Similar to the case of SADF, the use of the self-timed scheduling introduces timing interference between modes.
As a consequence, the analysis must take into account the sequence of mode transitions of interest.
Although the author provides an upper bound of timing behavior for a parameterized sequence of mode transitions, the accuracy is still unknown.
In contrast, our approach results in a timing analysis of mode transitions that is independent from already occurred transitions.
Moreover, the analysis results are directly reflected in the final implementation.
In this sense, our analysis is exact in the timing behavior of mode transitions.

In \cite{Bhattacharya2001}, a meta-modeling technique is proposed to augment the expressive power of wide range of existing data-flow models which have the graph iteration concept. In \cite{Bhattacharya2001}, the proposed technique is especially applied to the SDF model which is called Parameterized SDF (PSDF). In PSDF, separate \textit{init} and \textit{sub-init} graphs are proposed to reconfigure the body graph in a hierarchical manner. In this model, functional properties can only be partially decided at compile-time, and thus run-time verification is needed. To this end, for all configurations, computing a schedule and verifying consistency for both graphs and specifications need to be fulfilled at run-time which is pretty complex procedure. In addition, temporal analysis to find the worst-case system reconfiguration delay to preserve model predictability is not proposed. In contrast, our MADF model does not require run-time consistency check as every mode in our model is predefined at compile-time and represented as a CSDF graph. In addition, our MADF provides the temporal analysis of the mode transitions at compile-time using the MOO transition protocol.

In \cite{henia2007scenario,Negrean2011_etfa},
an analysis is proposed to reason about worst-case response time of
a task graph in case of a mode change. However, the task graph has very limited expressiveness and is not able to model the behavior of adaptive streaming applications. Instead, in our paper, we define a more expressive MoC that is amenable to adaptive application behavior and real-time analysis.

In~\cite{Real2004, Stoimenov2009_date}, the authors focus on timing analysis for mode changes of real-time tasks.
The starting times of new mode tasks need to be delayed to avoid overloading of processors during mode changes.
In \cite{Real2004, Stoimenov2009_date}, however, it is assumed that tasks are independent.
The proposed algorithms are thus not applicable to adaptive MoCs, since the starting times of tasks in adaptive MoCs depend on each other due to data dependencies.
Moreover, the algorithms in~\cite{Real2004, Stoimenov2009_date} involve high computational complexity because fixed-point equations must be solved at every step in the algorithms. In contrast, in our paper, we propose an adaptive MoC and analysis for applications with data-dependent tasks, which is more realistic and applicable to wider range of real-life streaming applications. Moreover, our analysis is simpler with low computational and time complexity.

 \section{Background}\label{background}
In this section, we provide a brief overview of our system model, the CSDF
MoC, and the scheduling framework presented in \cite{Bamakhrama2011_emsoft}. This background is needed to understand the novel contributions of our work.

\subsection{System Model}
The considered MPSoC platforms in this work are homogeneous, i.e., they may contain multiple, but the same type of programmable Processing Elements (PEs) with distributed memories.
Moreover, the platform must be predictable, which means timing guarantees are provided on the response time of hardware components and OS schedulers.
The precision-timed (PRET)~\cite{Lickly2008_pret} platform is such an example.
On the software side, we assume partitioned scheduling algorithms, i.e, no migration of tasks between PEs is allowed. The considered scheduling algorithms on each PE include Fixed-Priority Preemptive Scheduling (FPPS) algorithms, such as RM \cite{Liu1973_hrt}, or dynamic scheduling algorithms, such as EDF \cite{Liu1973_hrt}.
\subsection{Cyclo-Static Data Flow (CSDF)}\label{CSDF graph}
An application modeled as a CSDF \cite{Bilsen1996} is defined as a directed graph $G = (\mathcal{A},\mathcal{E})$ that consists of a set of actors $\mathcal{A}$ which communicate with each other through a set of edges $\mathcal{E}$. Actors represent computation while edges represent data dependency due to communication and synchronization. In CSDF, every actor $A_i \in \mathcal{A}$ has an execution sequence $C_i=[c_1, c_2, \dots , c_{\phi_i}]$ of length $\phi_i$. This means, the $x-$th time that actor $A_i$ is fired, it performs the computation $C_i(((x-1) \mod \phi_i )+1)$. Similarly, production and consumption of data tokens are also sequences of length $\phi_i$ in CSDF. The token production of actor $A_i$ to edge $E_j$ is represented as a sequence of constant integers $PRD_j=[prd_1,prd_2,\dots,prd_{\phi_i}]$, called \textit{production sequence}. Analogously, token consumption from every input edge $E_k$ of actor $A_i$ is a predefined sequence $CNS_k=[cns_1,cns_2,\dots,cns_{\phi_i}]$, called \textit{consumption sequence}. The $x-$th time that actor $A_i$ is fired, it produces $PRD_j(((x-1) \mod \phi_i) + 1)$ tokens to channel $E_j$ and consumes $CNS_k(((x-1) \mod \phi_i) + 1)$ tokens from channel $E_k$.

An important property of the CSDF model is the ability to derive a schedule for the actors at compile-time. In order to derive a valid static schedule for a CSDF graph at compile-time, it has to be consistent and live.
\begin{theorem}[From \cite{Bilsen1996}]\label{repetition}
	In a CSDF graph $G$, a repetition vector $\vec{q}=[q_1, q_2, \cdot\cdot\cdot, 
	q_{|\mathcal{A}|}]^T$ is given by
	\begin{equation}\label{repetition_vector}
		\small
	\vec{q}=\Theta\cdot\vec{r} ~~~~with~~~~
	\Theta_{j,i}=\begin{cases} \phi_i &\quad if~j = i\\
	0 &\quad otherwise\end{cases}
	\end{equation}\label{balance_equation}
	where 	 $\vec{r}=[r_1, r_2, ... , r_{|\mathcal{A}|}]^T$ is a positive integer solution of the balance equation $\Gamma\cdot\vec{r}= \vec{0}$ and where the topology matrix $\Gamma \in \mathbb{Z}^{|\mathcal{E}|\times|\mathcal{A}|}$ is defined by
	
$		\small\Gamma_{j,i}=\begin{cases}\sum_{k=1}^{k=\phi_i}PRD_j(k) & if~actor~A_i~produces~to~edge~E_j\\-\sum_{k=1}^{k=\phi_i}CNS_j(k) & if~actor~A_i~consumes~from~edge~E_j\\0 &\quad otherwise.\end{cases}$
\end{theorem}

A CSDF graph $G$ is said to be consistent if a positive integer solution $\vec{r}=[r_1, r_2, ... , r_{|\mathcal{A}|}]^T$ exists for the balance equation in \EQ~(\ref{balance_equation}). If a deadlock-free schedule can be found, $G$ is said to be live. Each consistent CSDF graph has a non-trivial repetition vector $\vec{q}=[q_1, q_2, \cdot\cdot\cdot, q_{|\mathcal{A}|}]^T\in \mathbb{N}^{|\mathcal{A}|}$. An entry $q_i \in \vec{q}$ denotes how many times an actor $A_i \in \mathcal{A}$ has to be executed in every graph iteration of $G$. For more details, we refer the reader to \cite{Bilsen1996}.

\subsection{Strictly Periodic Scheduling of CSDF}\label{SPS}
In \cite{Bamakhrama2011_emsoft}, a real-time strictly periodic scheduling (SPS) framework for CSDF graphs is proposed. In this framework, the actors in a CSDF graph are converted to a set of real-time implicit-deadline periodic tasks. Therefore, such a real-time task corresponding to a CSDF actor is associated with two parameters, namely period $T$ and earliest starting time $S$, where the deadline of the task is equal to its period (i.e., implicit deadline). 
The minimum period $T_i$ \cite{Bamakhrama2011_emsoft} of any actor $A_i\in \mathcal{A}$ under SPS can be computed as:
\begin{equation}\label{period}
T_i=\frac{lcm(\vec{q})}{q_i} \bigg\lceil \frac{\max_{A_i \in \mathcal{A}} \{\mu_i q_i\}}{lcm(\vec{q})}\bigg\rceil,\end{equation}
where $q_i$ is the number of repetitions of actor $A_i$ per graph iteration, and $\mu_i$ is the worst-case execution time (WCET) of actor $A_i$. In general, the derived period vector $\vec{T}$ must satisfy the condition $q_1T_1 = q_2T_2 = \cdot\cdot\cdot = q_nT_n = H$, where $H$ is the iteration period, also called hyper period, that represents the duration needed by the graph to complete one iteration. The minimum period of the sink actor for a CSDF graph determines the maximum throughout that this graph can achieve.
In addition, the utilization of any actor $A_i\in \mathcal{A}$, denoted by $u_i$, can be computed as $u_i = \mu_i/T_i$, where $u_i \in (0, 1]$.

\begin{figure}[!t!b]
	\centering
	\includegraphics[width=1\columnwidth]{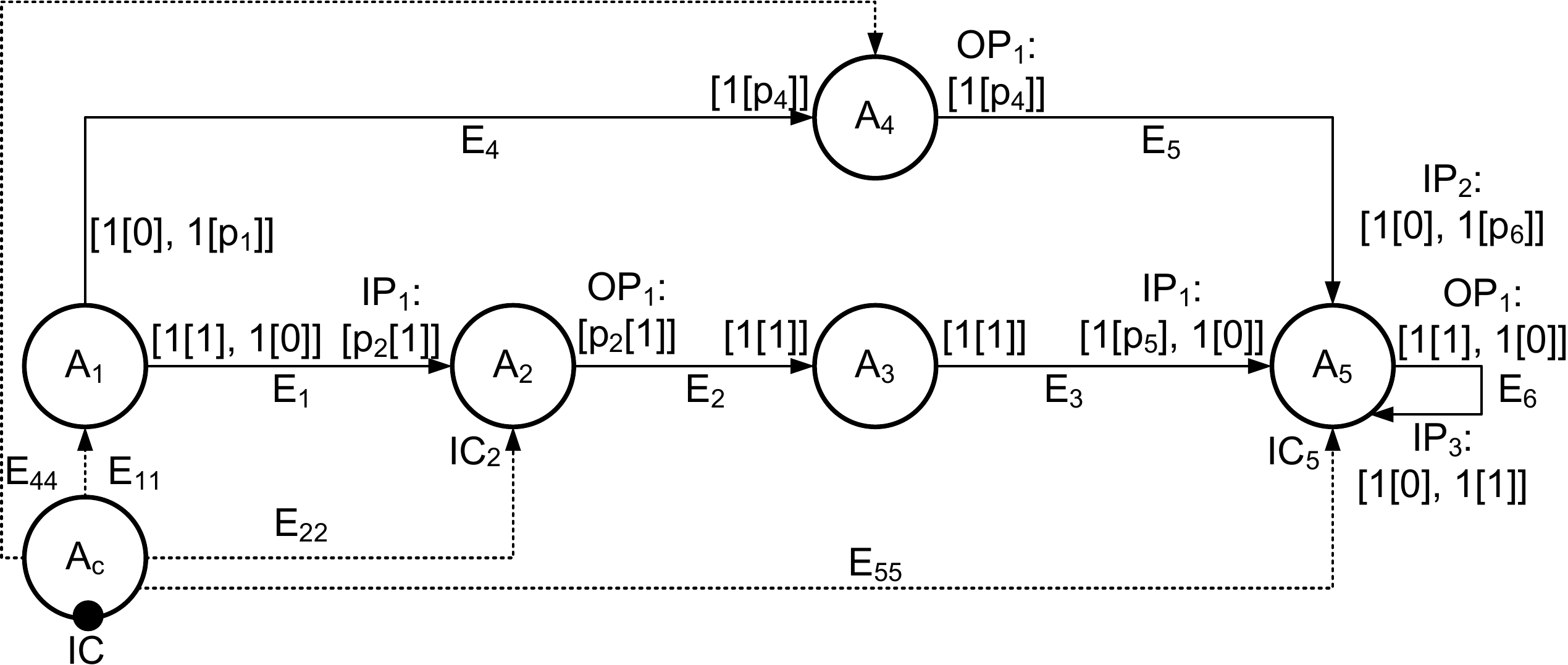}
	\caption{An example of \SADF{} graph ($G_1$).}
	\label{fig:sacsdf_ex}
\end{figure}
To sustain a strictly periodic execution with the period derived by \EQ~(\ref{period}), the earliest starting time $S_i$ \cite{Bamakhrama2011_emsoft} of any actor $A_i\in \mathcal{A}$ can be obtained as:
\begin{equation}\label{start_time}
S_i=\begin{cases} 0 &\quad if ~\pred(A_i) = \emptyset\\
\max_{A_j \in \pred(A_i)}(S_{j\rightarrow i}) &\quad otherwise, \end{cases}
\end{equation}
where $\pred(A_i)$ represents the set of predecessor actors of $A_i$ and $S_{j \rightarrow i}$ is given by:
\begin{equation}\label{continue_start_time}
\begin{split}
S_{j\rightarrow i}= &\underset{t \in[0,S_j+H]}{\min} \Big\{t: \underset{[S_j,\max\{S_j,t\}+k)}{\prd}(A_j,E_u)\\ &\geq \underset{[t,\max\{S_j,t\}+k]}{\cns}(A_i,E_u),~ \forall k\in[0,H],k\in\mathbb{N}\Big\}
\end{split}
\end{equation}
where $\prd_{[t_s,t_e)}(A_j,E_u)$ is the total number of tokens produced by $A_j$ to edge $E_u$ during the time interval $[t_s,t_e)$ and $\cns_{[t_s,t_e]}(A_i,E_u)$ is the total number of tokens consumed by $A_i$ from edge $E_u$ during the time interval $[t_s,t_e]$. \EQ~(\ref{continue_start_time}) considers the dependency between actors $A_j$ and $A_i$, over directed channel $E_u$. It calculates the earliest starting time $S_{j\rightarrow i}$ such that $A_i$ is never blocked on reading data tokens from $E_u$ during its periodic execution. This is ensured by checking that at each time instant, actor $A_i$ can be fired such that the cumulative number of tokens produced by $A_j$ over $E_u$ is greater than or equal to the number of tokens $A_i$ consumes from $E_u$. Start times $S_{j\rightarrow i}$ are computed for each actor $A_j$ in the predecessor set of $A_i$, i.e., $A_j \in\pred(A_i)$. Then, when actor $A_i$ has several predecessors, the earliest starting time $S_i$ has to be set to the maximum of starting times $S_{j\rightarrow i}$ considering each predecessor in isolation, as captured by \Cref{start_time}. For more details, we refer the reader to \cite{Bamakhrama2011_emsoft}.

 \section{Mode-Aware Data Flow~(\SADF)}\label{sec:dyn_moc_def}

In this section, we introduce our new MoC called Mode-Aware Data Flow (\SADF{}). \SADF{} can capture multiple modes associated with an adaptive streaming application, where each individual mode is a CSDF graph \cite{Bilsen1996}. Details and formal definitions of the \SADF{} model and its operational semantics are given later in this section. Here, we explain the \SADF{} intuitively by an example. Throughout this paper, we use graph $G_1$ shown in \FIG~\ref{fig:sacsdf_ex} as the running example to illustrate the definition of \SADF{} and the hard real-time scheduling analysis related to \SADF{}. This graph consists of 5 computation actors $A_1$ to $A_5$ that communicate data over edges $E_1$ to $E_5$. Also, there is an extra actor $A_c$ which controls the switching between modes through control edges $E_{11}$, $E_{22}$, $E_{44}$, and $E_{55}$ at run-time. Each edge contains a production and a consumption pattern, and some of these production and consumption patterns are parameterized. Having different values of parameters and worst-case execution times (WCET) of the actors determine different modes. For example, to specify the consumption pattern with variable length on edge $E_1$ in graph $G_1$, the parameterized notation $[p_2[1]]$ is used on edge $E_1$ that is interpreted as a sequence of $p_2$ elements with integer value $1$, e.g., $[2[1]] = [1,1]$. Similarly, the notation $[1[p_4]]$ on edge $E_4$ is interpreted as a sequence of $1$ element with integer value $p_4$, e.g., $[1[2]] = [2]$. Assume in this particular example that parameter vector ($p_1,p_2,p_4,p_5,p_6$) can take only two values (0, 2, 0, 2, 0) and (1, 1, 1, 1, 1). Then, $A_c$ can switch the application between two corresponding modes $\textit{SI}^1$ and $\textit{SI}^2$ by setting the parameter vector to value (0, 2, 0, 2, 0) and (1, 1, 1, 1, 1), respectively, at run-time. \FIG~\ref{fig:sacsdf_sn}(a) and (b) show the corresponding CSDF graphs of mode $\textit{SI}^1$ and $\textit{SI}^2$.
\begin{figure}[!t!b]
	\centering
	\subfigure[CSDF graph $G_1^1$ of mode $\textit{SI}^1$.]
	{\label{fig:sacsdf_sn1}\includegraphics[width=0.4\textwidth]{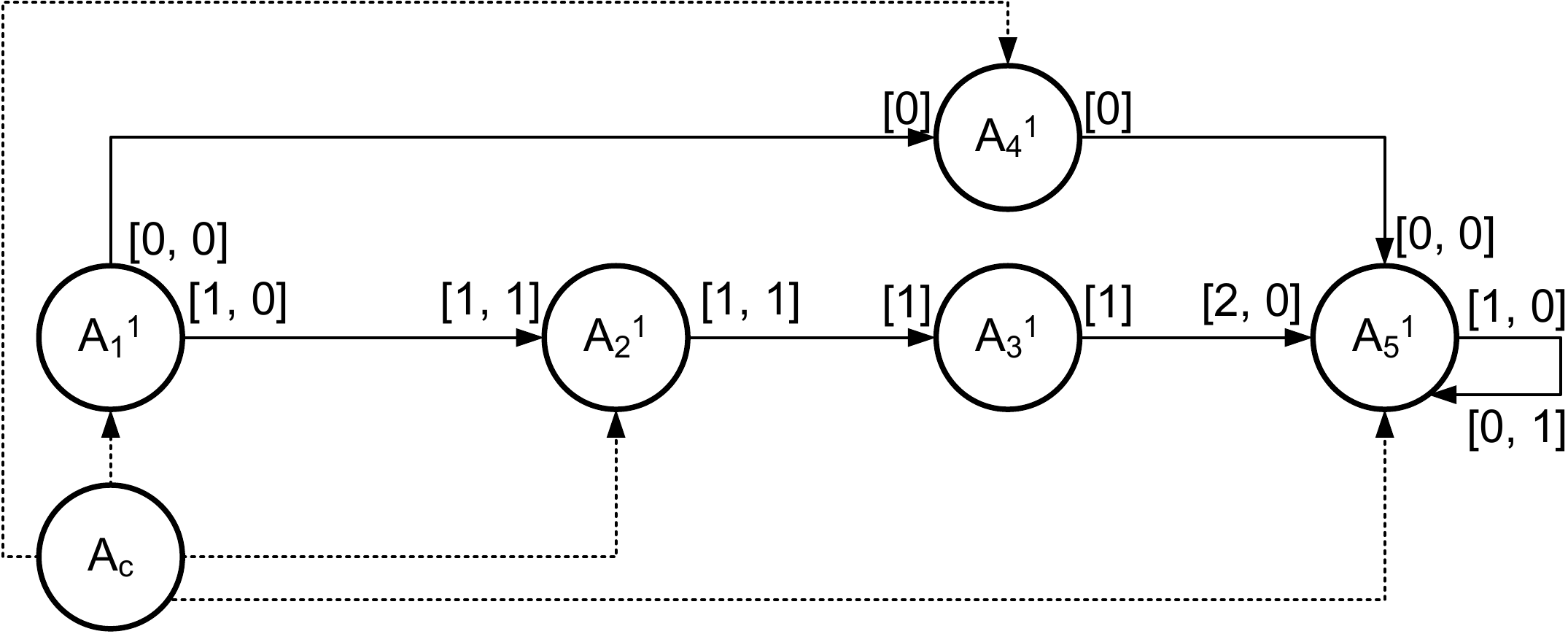}}
	\hspace{1mm}
	\subfigure[CSDF graph $G_1^2$ of mode $\textit{SI}^2$.]
	{\label{fig:sacsdf_sn2}\includegraphics[width=0.4\textwidth]{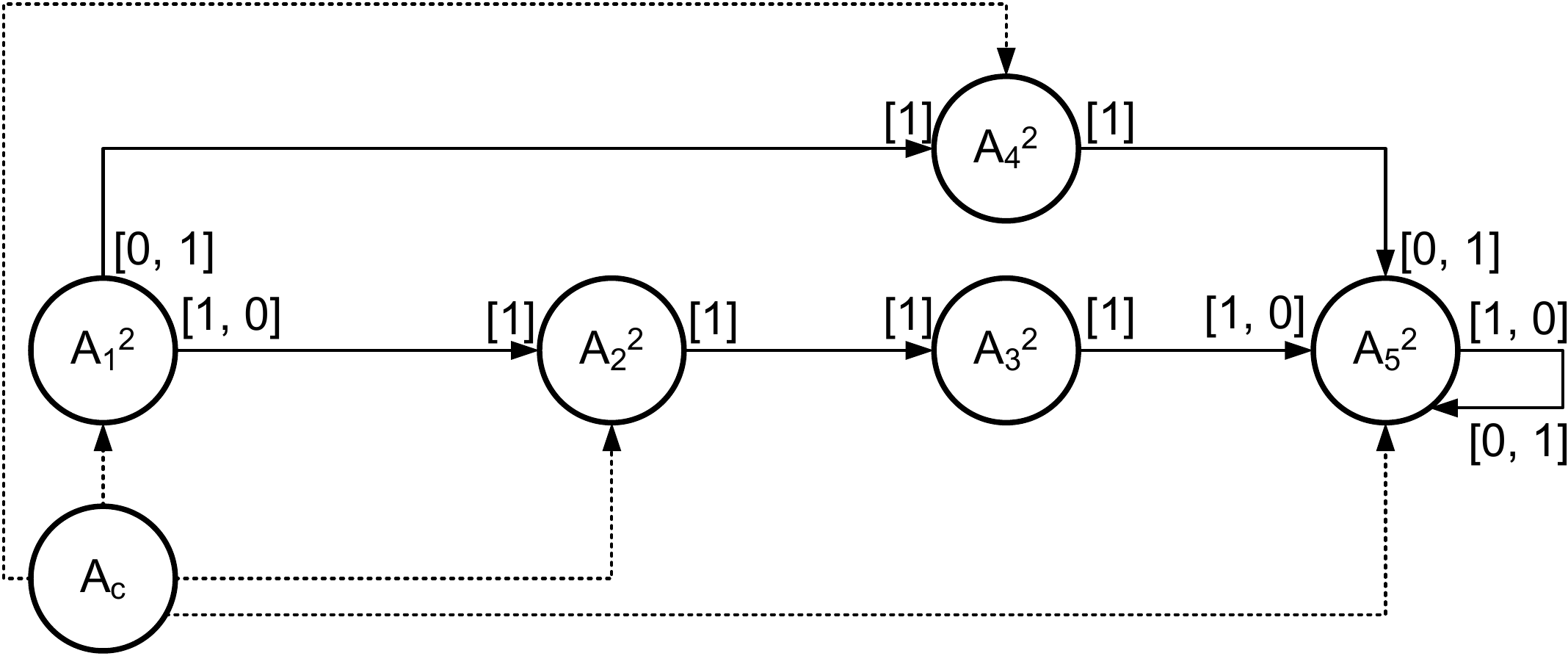}}
	\caption{Two modes of the \SADF{} graph in \FIG~\ref{fig:sacsdf_ex}.}
	\label{fig:sacsdf_sn}
\end{figure}
\subsection{Formal Definition of~\SADF}
\begin{definition}[Mode-Aware Data Flow~(\SADF)]\label{def:sacsdf}
	\index{Mode-Aware Data Flow~(\SADF)}\hfill \break
	A Mode-Aware Data Flow~(\SADF) is a multi-graph defined by a tuple $(\mathcal{A}, A_{c}, \mathcal{E},\Pi)$, where
	\begin{itemize}
		\item $\mathcal{A} = \{A_1, \dots, A_{|\mathcal{A}|}\}$ is a set of dataflow actors;
		\item $A_{c}$ is the {control} actor to determine modes and their transitions;
		\item $\mathcal{E}$ is the set of edges for data/parameter transfer;
		\item $\Pi = \{\vec{p}_1,\dots,\vec{p}_{|\mathcal{A}|}\}$ is the set of parameter vectors, where each $\vec{p}_i \in \Pi$ is
		associated with a dataflow actor $A_i$.
	\end{itemize}
\end{definition}

For $G_1$, $\mathcal{A} = \{A_1,A_2,A_3,A_4,A_5\}$ is the set of dataflow actors.
$A_{c}$ is the control actor.
$\mathcal{E} = \{E_1, E_2, E_3,E_4,E_5,E_6,E_{11}, E_{22}, E_{44}, E_{55}\}$ is the set of edges.
For actor $A_5$, $\vec{p}_5 = [p_5, p_6]$ is the parameter vector.
The input port $\textit{IP}_1$ of actor $A_5$ has a consumption sequence $[1[p_5], 1[0]]$, which can be interpreted as $[p_5,0]$.

\begin{definition}[Dataflow Actor]\label{def:d_a}
	A dataflow actor $A_i$ is descr\-ibed by a tuple $(\mathcal{I}_{i},\ {\textit{IC}}_\textit{i},\ \mathcal{O}_i, \mathcal{C}_i, M_i)$, where
	\begin{itemize}
		\item $\mathcal{I}_i = \{\textit{IP}_1,\dots,\textit{IP}_{|\mathcal{I}_i|}\}$ is the set of data input ports of actor $A_i$;
		\item ${\textit{IC}}_\textit{i}$ is the control input port that reads parameter vector $\vec{p}_i$ for actor $A_i$;
		\item $\mathcal{O}_i = \{\textit{OP}_1,\dots,\textit{OP}_{|\mathcal{O}_i|}\}$ is the set of data output ports of actor $A_i$;
		\item $\mathcal{C}_i = \{c_1, \dots, c_{|\mathcal{C}|}\}$ is the set of computations. When actor $A_i$ fires, it performs a
		computation $c_k \in \mathcal{C}_i$;
		\item ${M}_i~:~\vec{p}_i \rightarrow \{\phi, \bar{{C}}_i\}$ is a mapping relation, where $\vec{p}_i \in \Pi$, $\phi \in \mathbb{N}^+$, and
		$\bar{{C}}_i \subseteq C_i$ is a sequence of computations $[\bar{C}_i(1),\dots, \bar{{C}}_i(k),\dots,\bar{{C}}_i(\phi)]$ with
		$\bar{C}_i(k) \in \mathcal{C}_i, 1 \le k \le \phi$.
	\end{itemize}
\end{definition}
Actor $A_2$ in \FIG~\ref{fig:sacsdf_ex} has a set of one input port $\mathcal{I}_2 = \{\textit{IP}_1\}$, a set of one output port $\mathcal{O}_2 = \{\textit{OP}_1\}$ as well as a control input port $\textit{IC}_2$.
A set of computations $\mathcal{C}_2 = \{c_1, c_2, c_3\}$ is associated with $A_2$.
The mapping relation ${M}_2$ is given in \TAB~\ref{tab:ac_relation}.
\begin{table}[t]
	\centering
	\begin{minipage}{0.22\textwidth}
		\centering
		\caption{Mapping relation $M_2$ for actor $A_2$ in \FIG~\ref{fig:sacsdf_ex}.} \label{tab:ac_relation}
		\begin{tabular}{| c | c | c|}
			\hline
			$\vec{p}_2 = [p_2]$  & $\phi$       & \scriptsize{$\bar{C}_2$} \\ \hline \hline
			2                   & 2            & $[c_1,c_2]$            \\ \hline
			1                   & 1            & $[c_3]$                \\ \hline
		\end{tabular}
	\end{minipage}
	\hspace{1mm}
	\begin{minipage}{0.22\textwidth}
		\footnotesize
		\centering
	    \caption{Function $\textit{MC}_5$ defined for actor $A_5$ in \FIG~\ref{fig:sacsdf_ex}.}\label{tab:c_func_a5}
		\begin{tabular}{ |c | c | c | c | c | c | c | c |}
			\hline
			~~~ $\mathcal{S}$   ~~~~ & ~~~ $\mathbb{N}^2 $ ~~~ \\ \hline \hline
			~~~~ $\textit{SI}^1$ ~~~~ & ~~~~ $[2,0]$ ~~~~      \\ \hline
			~~~~ $\textit{SI}^2$ ~~~~ & ~~~~ $[1,1]$ ~~~~      \\ \hline
		\end{tabular}
	\end{minipage}
\end{table}
It can be interpreted as follows:
If $p_2 = 2$, actor $A_2$ repetitively performs computations according to sequence $\bar{C}_2 = [c_1,c_2]$ every time when firing $A_2$.
When $p_2 = 1$, firing $A_2$ performs computation $c_3$.

\begin{definition}[Control Actor]\label{def:a_c}
	The control actor $A_{c}$ is descr\-ibed by a tuple $({\textit{IC}}, \mathcal{O}_c, \mathcal{S}, \mathcal{M}_c)$, where
	\begin{itemize}
		\item $\mathcal{S} = \{\textit{SI}^1,\dots,\textit{SI}^{|\mathcal{S}|}\}$ is a set of mode identifiers, each of which
		specifies a unique mode;
		\item $\textit{IC}$ is the control input port which is connected to the external environment.
		Mode identifiers are read through the control input port from the environment;
		\item $\mathcal{O}_c = \{\textit{OC}_1, \dots, \textit{OC}_{|\mathcal{A}|}\}$ is a set of control output ports.
		Parameter vector $\vec{p}_i$ is sent through $\textit{OC}_i \in \mathcal{O}_c$ to actor $A_i$;
		\item $\mathcal{M}_c = \{\textit{MC}_1,\dots,\textit{MC}_{|\mathcal{A}|}\}$ is a set of functions defined for each actor $A_i \in \mathcal{A}$.
		For each $\textit{MC}_i \in \mathcal{M}_c$, $\textit{MC}_i~:~\mathcal{S} \rightarrow \mathbb{N}^{|\vec{p}_i|}$
		is a function that takes a mode identifier and outputs a vector of non-negative integer values.
	\end{itemize}
\end{definition}
For $G_1$ in \FIG~\ref{fig:sacsdf_ex}, we have two mode identifiers $\mathcal{S} = \{\textit{SI}^1,\textit{SI}^2\}$.
At run-time, control actor $A_c$ reads these mode identifiers through control port \textit{IC} (black dot in \FIG~\ref{fig:sacsdf_ex}).
For actor $A_5$, $\textit{MC}_5 \in \mathcal{M}_c$ is given in \TAB~\ref{tab:c_func_a5}.
As explained previously, the parameter vector for actor $A_5$ is $\vec{p}_5 = [p_5, p_6]$.
Therefore, $\textit{MC}_5$ takes a mode identifier and outputs a 2-dimensional vector as shown in the second column in \TAB~\ref{tab:c_func_a5}.
For instance, mode $\textit{SI}^1$ results in a non-negative integer vector $[2,0]$.

To further define production/consumption sequences with variable length, we use the notation $n [m]$
for a sequence of $n$ elements with integer value $m$, i.e., 
\begin{align*}n [m] = [\overbrace{m,\dots,m}^ {n \text{ times }}].\end{align*}

\begin{definition}[Input Port]\label{def:sacsdf_ip}
	An {input port} {\textit{IP}} of an actor is described by a tuple (\textit{CNS}, $M_\textit{IP}$), where
	\begin{itemize}
		\item $\textit{CNS} = [\phi_1 [\textit{cns}_1],\dots,  \phi_K [\textit{cns}_K]]$ is the consumption sequence with $\phi$ phases,
		where $\phi = \sum_{i=1}^K \phi_i$ is determined by the mapping relation $M$ in \DEF~\ref{def:d_a}, and $\textit{cns}_1,\dots,\textit{cns}_K \in \mathbb{N}$;
		\item $M_\textit{IP}~:~\vec{p}_i \rightarrow \psi_\textit{IP}$ is a mapping relation, where $\vec{p}_i \in \Pi$
		and
		\begin{equation}\label{eq:sacsdf_ip}
		\psi_\textit{IP} = \{\phi_1,\dots,\phi_K,\textit{cns}_1,\dots,\textit{cns}_K\}.
		\end{equation}
	\end{itemize}
\end{definition}
\begin{definition}[Output Port]\label{def:sacsdf_op}
	An {output port} \textit{OP} of an actor is described by a tuple ($\textit{PRD}, M_\textit{OP}$), where
	\begin{itemize}
		\item $\textit{PRD}=[\phi_1 [\textit{prd}_1],\dots,  \phi_K [\textit{prd}_K]]$ is the production sequence with $\phi$ phases,
		where $\phi = \sum_{i=1}^K \phi_i$ is determined by the mapping relation $M$ in \DEF~\ref{def:d_a}, and $\textit{prd}_1,\dots,\textit{prd}_K \in \mathbb{N}$.
		\item $M_\textit{OP}~:~\vec{p}_i \rightarrow \psi_\textit{OP}$ is mapping relation,
		where $\vec{p}_i \in \Pi$ and
		\begin{equation}\label{eq:sacsdf_op}
		\psi_\textit{OP} = \{\phi_1,\dots,\phi_K,\textit{prd}_1,\dots,\textit{prd}_K\}.
		\end{equation}
	\end{itemize}
\end{definition}

The consumption/production sequence defined here is a generalization of that for the CSDF MoC (see \Cref{CSDF graph}).
We can see that a CSDF actor has a constant $\phi$ phases in its consumption/production sequences, whereas the length of the phase of an \SADF{} actor is parameterized by $\phi = \sum_{i=1}^K \phi_i$.
In addition, the mapping relation $M_\textit{IP}/M_\textit{OP}$ must be provided by the application designer.
Consider the two input ports $\textit{IP}_1$ and $\textit{IP}_2$ of actor $A_5$ in \FIG~\ref{fig:sacsdf_ex}.
The mapping relations $M_{\textit{IP}_1}$ and $M_{\textit{IP}_2}$ are represented as follows:
\begin{equation}
\hspace{-0.4mm}
M_{\textit{IP}_1}~:~\vec{p}_5 = [p_5, p_6] \rightarrow\psi_{\textit{IP}_1} = \{ \phi_1,\phi_2,\textit{cns}_1, \textit{cns}_2 \} = \{ 1,1,p_5,0 \}, \label{eq:m_ip1}
\end{equation}
\begin{equation}
\hspace{-0.4mm}
M_{\textit{IP}_2}~:~\vec{p}_5 =[p_5, p_6] \rightarrow \psi_{\textit{IP}_2} = \{ \phi_1,\phi_2,\textit{cns}_1, \textit{cns}_2 \} = \{ 1,1,0,p_6 \}. \label{eq:m_ip2}
\end{equation}
It can be seen that parameter $p_5$ is mapped to $\textit{cns}_1$ of $\textit{IP}_1$,
parameter $p_6$ is mapped to $\textit{cns}_2$ of $\textit{IP}_2$, and $\phi_1$ and $\phi_2$ both are constant equal to 1.
Therefore, the consumption sequence of $\textit{IP}_1$ is $\textit{CNS} = [1[p_5], 1[0]]=[p_5,0]$
and the consumption sequence of $\textit{IP}_2$ is $\textit{CNS} = [1[0], 1[p_6]]=[0,p_6]$.
Similarly considering output port $\textit{OP}_1$ of actor $A_4$, its mapping relation $M_{\textit{OP}_1}$ is given as:
\begin{equation}
M_{\textit{OP}_1}~:~\vec{p}_4 = [p_4] \rightarrow \psi_{\textit{OP}_1} = \{ \phi_1,\textit{prd}_1 \} = \{ 1,p_4\}.
\end{equation}
In this case, parameter $p_4$ is mapped to $\textit{prd}_1$ and $\phi_1 = 1$.
Therefore, production sequence $\textit{PRD} = [1[p_4]]=[p_4]$ is obtained for $\textit{OP}_1$ of $A_4$.

\begin{definition}[Edge]
	An {edge} $E \in \mathcal{E}$ is defined by a tuple $\Big((A_i,\textit{OP}), (A_j, \textit{IP})\Big),$ where
	\begin{itemize}
		\item actor $A_i$ produces a parameterized number of tokens to edge $E$ through output port $\textit{OP}$;
		\item actor $A_j$ consumes a parameterized number of tokens from $E$ through input port $\textit{IP}$.
	\end{itemize}
\end{definition}
Considering edge $E_5$ in \FIG~\ref{fig:sacsdf_ex}, it connects output port $\textit{OP}_1$ of actor $A_4$ to input port $\textit{IP}_2$ of actor $A_5$.

\begin{figure}[!t!b]
	\centering
	\subfigure[Mode $\textit{SI}^1$ in \FIG~\ref{fig:sacsdf_sn1}.]
	{\label{fig:static_ex1}\includegraphics[width=0.235\textwidth]{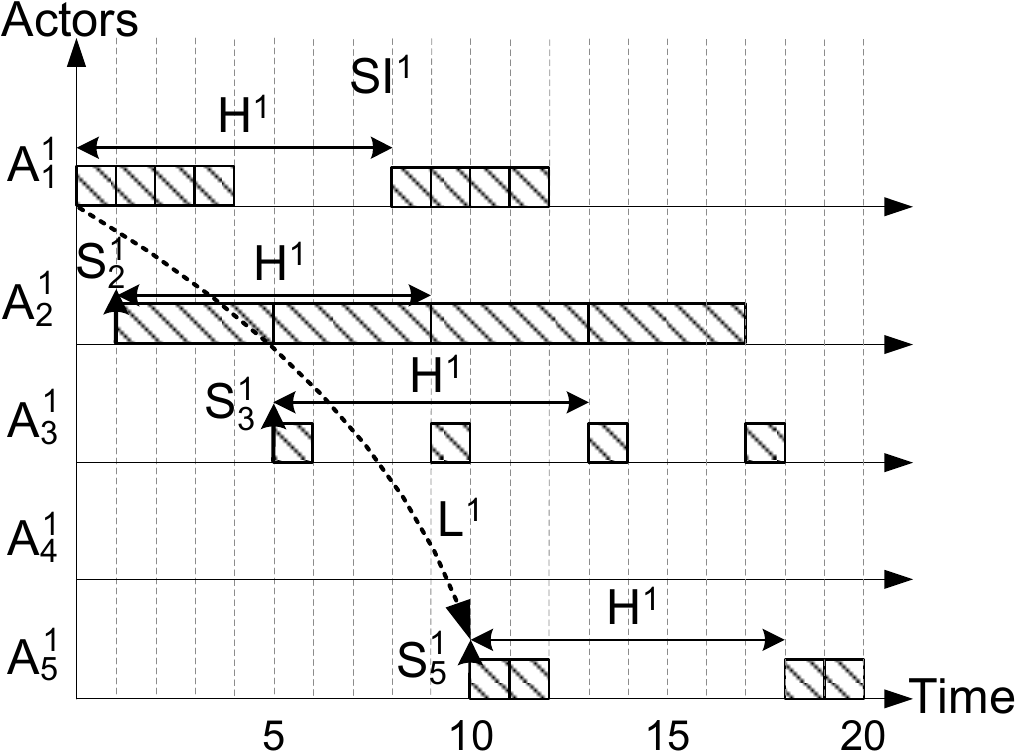}}
	\subfigure[Mode $\textit{SI}^2$ in \FIG~\ref{fig:sacsdf_sn2}.]
	{\label{fig:static_ex2}\includegraphics[width=0.235\textwidth]{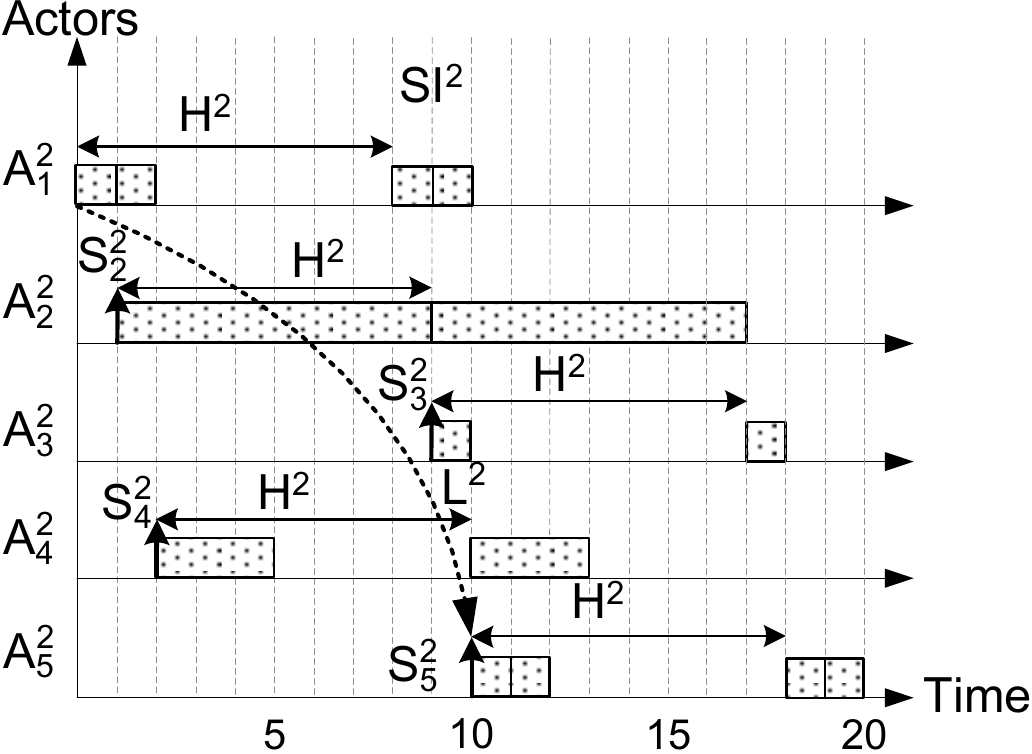}}
	\caption{Execution of two iterations of both modes $\textit{SI}^1$ and $\textit{SI}^2$ under self-timed scheduling.}
	\label{fig:static_ex}
\end{figure}
\begin{definition}[Mode of \SADF]\label{def:sacsdf_sn}
	A mode $\textit{SI}^i$ of \SADF{} is a consistent and live CSDF graph, denoted as $G^i$, obtained by setting values of $\Pi$ in \DEF~\ref{def:sacsdf} as follows:	
	\begin{equation}\label{eq:sacsdf_sn}
	\forall \vec{p}_k \in \Pi~:~ \vec{p}_k = \mc_k(\textit{SI}^i),
	\end{equation}
	where function $\textit{MC}_k$ is given in \DEF~\ref{def:a_c}.
\end{definition}
\begin{definition}[Mode of \SADF{} Actor]\label{def:sacsdf_sn_a}
	An actor $A_k$ in mode $\textit{SI}^i$, denoted by $A_k^i$, is a CSDF actor obtained from $A_k$ as follows:
	\begin{equation}\label{eq:a_param_sn}
	\vec{p}_k = \mc_k(\textit{SI}^i).
	\end{equation}
\end{definition}

\FIG~\ref{fig:sacsdf_sn1} shows the CSDF graph of mode $\textit{SI}^1$ and  \FIG~\ref{fig:sacsdf_sn2} shows the CSDF graph
of mode $\textit{SI}^2$.
Consider function $\textit{MC}_5$ for actor $A_5$ in \TAB~\ref{tab:c_func_a5} with parameter vector $\vec{p}_5 = [p_5,p_6]$.
For instance, mode $\textit{SI}^1$ results in $\vec{p}_5 = [p_5,p_6] = [2,0]$, where parameter values $p_5 = 2$ and $p_6=0$.
Consequently, according to mapping relations $M_{\textit{IP}_1}$ and $M_{\textit{IP}_2}$ given in \Cref{eq:m_ip1} and \Cref{eq:m_ip2}, $\textit{cns}_1 = p_5 = 2$ can be obtained for input port $\textit{IP}_1$ and $\textit{cns}_2 = p_6 = 0$ for $\textit{IP}_2$.
This determines actor $A_5^1$ shown in  \FIG~\ref{fig:sacsdf_sn1} for mode $\textit{SI}^1$.

\begin{definition}[Inactive Actor]
	An \SADF{} actor $A_i^k$ is inactive in mode $\textit{SI}^k$ if the following conditions hold:
	\begin{enumerate}
		\item $\forall \textit{IP} \in \mathcal{I}_i~:~ \textit{CNS} = [0,\dots,0]$;
		\item $\forall \textit{OP} \in \mathcal{O}_i~:~ \textit{PRD} = [0,\dots,0]$.
	\end{enumerate}
	Otherwise, $A_i^k$ is called \textit{active} in mode $\textit{SI}^k$.
\end{definition}

For actor $A_4^1$ shown in \FIG~\ref{fig:sacsdf_sn1}, it has consumption and production sequence $[0]$.
Therefore, actor $A_4$ is said to be inactive in mode $\textit{SI}^1$.

\subsection{Operational Semantics}\label{sec:sacsdf_opseman}
During execution of a \SADF{} graph, it can be either in a steady-state or mode transition.

\begin{definition}[Steady-state]\label{def:sacsdf_steady}
	A \SADF{} graph is in a steady-state of a mode $\textit{SI}^i$, if it satisfies \Cref{eq:sacsdf_sn} with the same $\textit{SI}^i$ for all its actors.
\end{definition}

\begin{definition}[Mode Transition]\label{def:sacsdf_trans}
	A \SADF{} graph is in a {mode transition} from mode $\textit{SI}^o$ to $\textit{SI}^l$, where $o \neq l$, if
	some actors have $\textit{SI}^o$ for \Cref{eq:a_param_sn} and the remaining active actors have $\textit{SI}^l$ for \Cref{eq:a_param_sn}.
\end{definition}

\begin{table}[!tb]
	\centering
	\caption{Actor parameter for $G_1$ in \FIG~\ref{fig:sacsdf_ex}.}
	\resizebox{0.47\textwidth}{!}{		\begin{tabular}{|c| c| c|c|c|c|c|c|c|c|c|}
			\hline
			Mode & \multicolumn{4}{|c|}{$\textit{SI}^1$} & \multicolumn{5}{|c|}{$\textit{SI}^2$}        \\ \hline \hline
			Actor & $A_1^1$ & $A_2^1$ & $A_3^1$ & $A_5^1$ & $A_1^2$ & $A_2^2$ & $A_3^2$ & $A_4^2$ & $A_5^2$ \\ \hline \hline
			WCET ($\mu_i$) & 1 & 4 & 1 & 1 			& 1 & 8 & 1 & 3 & 1 \\ \hline
			period ($T_i$) & 2 & 4 & 4 & 4 			& 4 & 8 & 8 & 8 & 4 \\ \hline
			starting time ($S_i$)      & 0 & 2 & 6 & 14 			& 0 & 4 & 12 & 8 & 20  \\\hline
			utilization ($u_i$)      & $\frac{1}{2}$ & $1$ & $\frac{1}{4}$ & $\frac{1}{4}$ 			& $\frac{1}{4}$ & $1$ & $\frac{1}{8}$ & $\frac{3}{8}$ & $\frac{1}{4}$  \\\hline
	\end{tabular}}
	\label{tab:sacsdf_param_ex2}
\end{table}
In the steady-state of a \SADF{} graph, all active actors execute in the same mode.
As defined previously in \Cref{def:sacsdf_sn} and shown in \FIG~\ref{fig:sacsdf_sn1} and \FIG~\ref{fig:sacsdf_sn2}, the steady-state of the \SADF{} graph has the same operational semantics as a CSDF graph.
We use $ \langle A_i^k, x \rangle$ to denote the $x$-th firing of actor $A_i$ in mode $\textit{SI}^k$.
At $ \langle A_i^k, x \rangle$, it executes computation $\bar{C}_i\big( ((x-1)\mod \phi) + 1 \big)$,
where $\bar{C}_i$ is given in \DEF~\ref{def:d_a}.
The number of tokens consumed and produced are specified according to Definitions~\ref{def:sacsdf_ip} and~\ref{def:sacsdf_op},
respectively.
For instance, the $x$-th firing of $A_i^k$ produces $\textit{PRD}\big( ((x-1) \mod \phi) + 1 \big)$ tokens through an output port $\textit{OP}$.
In each mode $\textit{SI}^k$, the \SADF{} graph is a consistent and live CSDF graph and thus has the notion of graph iterations with a non-trivial repetition vector $\vec{q}^k \in \mathbb{N}^{|\mathcal{A}|}$ resulting from \Cref{repetition_vector}.
Next, we further define mode iterations.
\begin{definition}[Mode Iteration]\label{def:sacsdf_rep}
	One iteration $\textit{It}^k$ of a \SADF{} graph in mode $\textit{SI}^k$ consists of
	one firing of control actor $A_{c}$ and
	$q_i^k \in \vec{q}^k$ firings of each \SADF{} actor $A_i^k$.
\end{definition}

Consider the two modes shown in \FIG~\ref{fig:sacsdf_sn1} and \FIG~\ref{fig:sacsdf_sn2}.
Repetition vectors $\vec{q}^1$ and $\vec{q}^2$ are:
\begin{align}
\vec{q}^1 = [4,2,2,0,2],  ~
\vec{q}^2 = [2,1,1,1,2]. \label{eq:sacsdf_q2}
\end{align}

For any mode of a \SADF{} graph, i.e., a live CSDF graph, under \textit{any} valid schedule, it has (eventually) periodic execution in time.
This holds for CSDF graphs under self-timed schedule~\cite{Stuijk2008_tc}, K-periodic schedule~\cite{Bodin2013_estimedia}, and SPS~\cite{Bamakhrama2011_emsoft}.
The length of the periodic execution, called \textit{iteration period}\index{iteration period~($H$)}, determines the minimum time interval to complete one graph iteration~(\textit{cf}. \DEF~\ref{def:sacsdf_rep}).
The iteration period, denoted by $H^k$, is equal for any actor in the same mode $\textit{SI}^k$.
During a periodic execution, the starting time of each actor $A_i^k$, denoted by $S_i^k$, indicates the time distance between
the start of source actor $A_\text{src}^k$ and the start of actor $A_i^k$ in the same iteration period.
Based on the notion of starting times, we define \textit{iteration latency} $L^k$ of a \SADF{} graph in mode $\textit{SI}^k$ as follows:
\begin{equation}\label{eq:iter_latency}
L^k = S_\text{snk}^k - S_\text{src}^k,
\end{equation}
where $S_\text{snk}^k$ and $S_\text{src}^k$ are the earliest starting times of the sink and source actors, respectively.
\FIG~\ref{fig:static_ex} illustrates the execution of both modes $\textit{SI}^1$ and $\textit{SI}^2$ given in \FIG~\ref{fig:sacsdf_sn} under the self-timed schedule.
A rectangle denotes the WCET of an actor firing.
The WCETs of all actors in both modes are given in the third row of \TAB~\ref{tab:sacsdf_param_ex2}.
Now, it can be seen in \FIG~\ref{fig:static_ex} that iteration period $H^1 = H^2 = 8$.
Based on the starting time of each actor, we obtain iteration latencies $L^1 = S_5^1 - S_1^1 = 10 - 0 = 10$ and $L^2 = S_5^2 - S_1^2= 10 - 0 = 10$ as shown in \FIG~\ref{fig:static_ex}.

\subsection{Mode Transition}\label{sec:sn_trans}
While the operational semantics of a \SADF{} graph in steady-state are the same as that of a CSDF graph,
the transition of \SADF{} graph from one mode to another is the crucial part that makes it fundamentally different from CSDF.
The protocol for mode transitions has strong impact on the compile-time analyzability and implementation efficiency.
In this section, we propose a novel and efficient protocol of mode transitions for \SADF{} graphs.

During execution of a \SADF{} graph, mode transitions may be triggered at run-time by receiving a
Mode Change Request (MCR) from the external environment.
We first assume that a MCR can be only accepted in the steady-state of a \SADF{} graph, not in an
ongoing mode transition.
This means that any MCR occurred during an ongoing mode transition will be ignored.
Consider a mode transition from $\textit{SI}^o$ to $\textit{SI}^l$.
The transition is accomplished by the control actor reading mode identifier $\textit{SI}^l$ from its control input port (see the black dot in \FIG~\ref{fig:sacsdf_ex})
and writing parameter values of $\vec{p}_i$ to the control output port connected to each dataflow actor $A_i^l$ according to function $\textit{MC}_i$ given in \DEF~\ref{def:a_c}.
Then, $A_i^l$ reads new parameter values $\vec{p}_i$ from its control input port and
sets the sequence of computations according to mapping relation $M_i$ in  \DEF~\ref{def:d_a}.
The production and consumption sequences are obtained in accordance with $M_\textit{IP}$ and $M_\textit{OP}$ in \DEF~\ref{def:sacsdf_ip} and \DEF~\ref{def:sacsdf_op}, respectively.
We further define/require that
mode transitions are only allowed at quiescent points~\cite{Neuendorffer2004}.
\begin{figure}[!t!b]
	\centering
	\includegraphics[width=1\columnwidth]{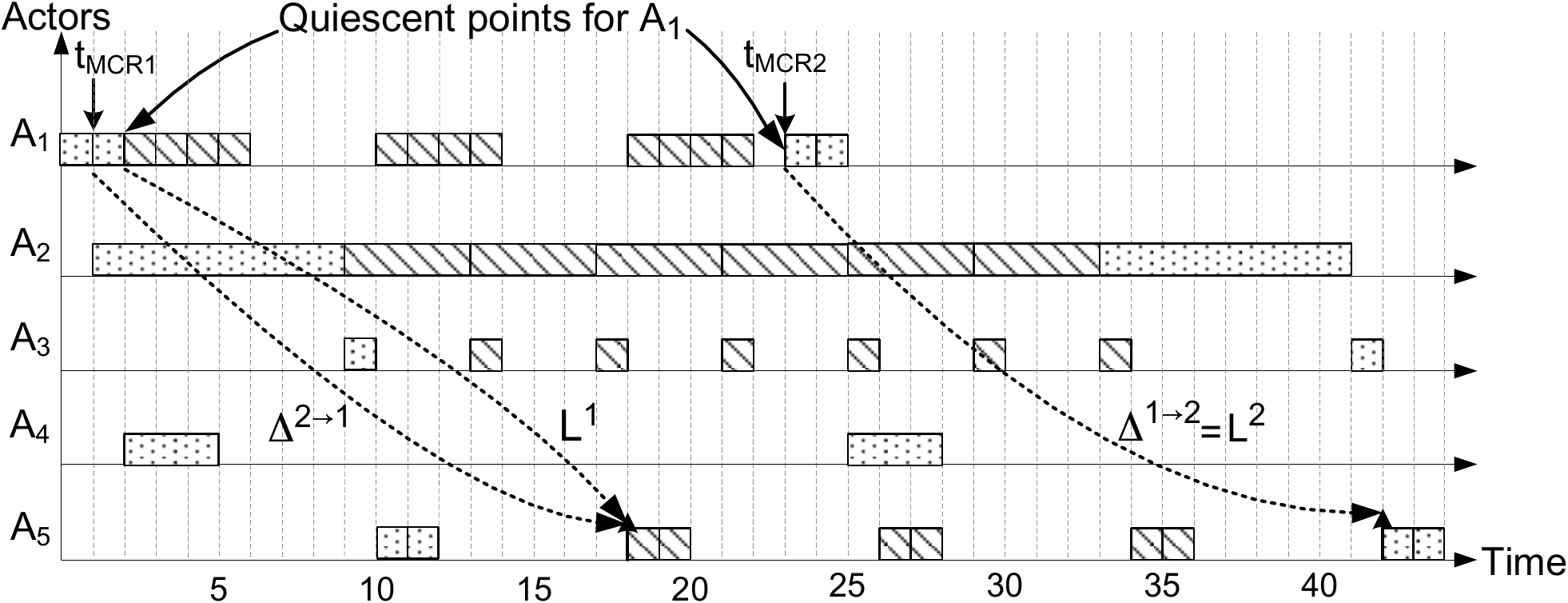}
	\caption{An execution of $G_1$ in \FIG~\ref{fig:sacsdf_ex} with two mode transitions under the ST transition protocol.
		\textit{MCR1} at time $t_\textit{MCR1}$ denotes a transition request from mode $\textit{SI}^2$ to $\textit{SI}^1$,
		and \textit{MCR2} at time $t_\textit{MCR2}$ denotes a transition request from mode $\textit{SI}^1$ to $\textit{SI}^2$.}
	\label{fig:st_transition}
\end{figure}

\begin{definition}[Quiescent Point of \SADF]\label{def:sacsdf_qp}
	For mode $\textit{SI}^l$, a quiescent point of \SADF{} actor $A_i$ is firing $\langle A_i^l, x \rangle$ in mode iteration $\textit{It}^l$ that satisfies
	\begin{equation}\label{eq:sacsdf_qp}
	\neg\exists \langle A_i^l, y \rangle \in \textit{It}^l ~:~y < x.
	\end{equation}
\end{definition}
\DEF~\ref{def:sacsdf_qp} simply refers to the first firing of actor $A_i$ in each iteration $\textit{It}^l$ of mode $\textit{SI}^l$. Recall that each iteration of mode $\textit{SI}^l$ consists of $q^l_i$ firings of actor $A_i$. Therefore, our requirement that a mode transition is only allowed at a quiescent point implies that a transition from mode $\textit{SI}^l$ to $\textit{SI}^o$ of actor $A_i$ happens when all firings of actor $A_i$ are completed in the iteration of $\textit{SI}^l$ when MCR occurs.
\FIG~\ref{fig:st_transition} shows an execution of $G_1$ in \FIG~\ref{fig:sacsdf_ex} with two mode transitions.
For instance, the MCR at time $t_\text{MCR1} = 1$ denotes a transition request from mode $\textit{SI}^2$ to $\textit{SI}^1$.
The mode transition of actor $A_1$ happens when all firings of actor $A_1$ are completed, that is at time 2 in \FIG~\ref{fig:st_transition} in this particular example.

\DEF~\ref{def:sacsdf_qp} defines mode transitions of \SADF{} graphs as partially ordered actor firings.
However, it does not specify at which time instance a mode transition actually starts.
Therefore, below, we focus on the transition protocol that defines the points in time for occurrences of mode transitions.
To quantify the transition protocol, we introduce a metric, called \textit{transition delay}, to measure the responsiveness of a protocol to a MCR.
\begin{definition}[Transition Delay\index{Transition delay}]\label{def:trans_delay}
	For a MCR at time $t_\text{MCR}$ calling for a
	mode transition from mode $\textit{SI}^o$ to $\textit{SI}^l$,
	the transition delay $\Delta^{o \rightarrow l}$ of a \SADF{} graph is defined as
	\begin{equation}\label{eq:trans_delay}
	\Delta^{o \rightarrow l}= \sigma_\text{snk}^{o \rightarrow l} - t_\text{MCR},
	\end{equation}
	where $\sigma_\text{snk}^{o \rightarrow l}$ is the earliest starting time of the sink actor in the new mode $\textit{SI}^l$.
\end{definition}

In \FIG~\ref{fig:st_transition}, we can compute the transition delay for \textit{MCR1} occurred at time $t_\text{MCR1} = 1$ as $\Delta^{2\rightarrow 1} = 18 - 1 = 17$.

\subsubsection{Self-timed Transition Protocol}\label{sec:st_trans}
In the existing adaptive MoCs like FSM-SADF~\cite{Geilen2010_codes10}, a protocol, referred here as \textit{Self-Timed}~(ST) transition protocol, is adopted.
The ST protocol specifies that actors are scheduled in the self-timed manner not only in the steady-state, but also during a mode transition.
For FSM-SADF upon a MCR, a firing of a FSM-SADF actor in the new mode can start immediately after the firing of the actor completes the old mode iteration.
The only possible delay is introduced due to availability of input data.
One reason behind the ST protocol is that the ST schedule for a (C)SDF graph (steady-state of \FSADF\footnote{The steady-state of SADF is defined similarly to that of \SADF. The only difference is that a scenario of \FSADF{} is a SDF graph, whereas a mode of \SADF{} is a CSDF graph.}) leads to
its highest achievable throughput.
However, the ST protocol generally introduces interference of one mode execution with another one.
The time needed to complete mode transitions also fluctuates
as the transition delay of an ongoing transition depends on the transitions that occurred in the past.
We consider this as an undesired effect because mode transitions using the ST protocol become potentially slow and unpredictable.
Another consequence of the incurred interference between modes using the ST transition protocol is the high time complexity of analyzing transition delays, because transition delays cannot be analyzed independently for each mode transition.
The analysis proposed in \cite{Geilen2010_codes10} uses an approach based on state-space exploration, which has the exponential time complexity.

Consider $G_1$ in \FIG~\ref{fig:sacsdf_ex}
and an execution of $G_1$ with the two mode transitions illustrated in \FIG~\ref{fig:st_transition}.
The execution is assumed under the ST schedule for both steady-state and mode transitions of $G_1$.
After \textit{MCR1} at time $t_\text{MCR1}$, the transition from mode $\textit{SI}^2$ to $\textit{SI}^1$ introduces interference to execution of the new mode $\textit{SI}^1$ from execution of the old mode $\textit{SI}^2$.
The interference increases the iteration latency of the new mode $\textit{SI}^1$ to $L^1 = S_5^1 - S_1^1 =  18-2 = 16$ from initially 10 as shown in \FIG~\ref{fig:static_ex1}
when $G_1$ is only executed in the steady-state of mode $\textit{SI}^1$.
Even worse, the interference is further propagated to the second mode transition after \textit{MCR2} at time $t_\text{MCR2}$.
In this case, the iteration latency $L^2 = S_5^2 - S_1^2 = 42 - 23 = 19$ is increased from initially 10 as shown in \FIG~\ref{fig:static_ex2}
when $G_1$ is only executed in the steady-state of mode $\textit{SI}^2$.
This example thus clearly shows the problem of the ST protocol.
That is, it introduces interference between the old and new modes due to mode transitions,
thereby increasing the iteration latency of the new mode in the steady-state after the transition.
Furthermore, the increase of iteration latency also potentially increases transition delays as it will be shown in the next section.

\subsubsection{Maximum-Overlap Offset Transition Protocol}\label{sec:moo_trans}
To address the problem of the ST transition protocol explained above, we propose a new transition protocol, called \textit{Maximum-Overlap Offset}~(MOO).
\begin{definition}[Maximum-Overlap Offset~(MOO)\index{Maximum-Overlap Offset~(MOO)}]\label{def:moo}
	For a \SADF{} graph and a transition from mode $\textit{SI}^o$ to $\textit{SI}^l$, Maximum-Overlap Offset~(MOO), denoted by $x$, is defined as
	\begin{equation}\label{eq:moo_x}
	x = \begin{cases}
	\max_{A_i \in \mathcal{A}^o \cap \mathcal{A}^l}(S_i^o - S_i^l)       & \mbox{if  } \max_{A_i \in \mathcal{A}^o \cap \mathcal{A}^l}(S_i^o - S_i^l) > 0 \\
	0                                                   & \mbox{otherwise},
	\end{cases}
	\end{equation}
	where $\mathcal{A}^o \cap \mathcal{A}^l$ is set of actors active in both modes $\textit{SI}^o$ and $\textit{SI}^l$.
\end{definition}

\begin{figure}[!t!b]
	\centering
	\includegraphics[width=.7\columnwidth]{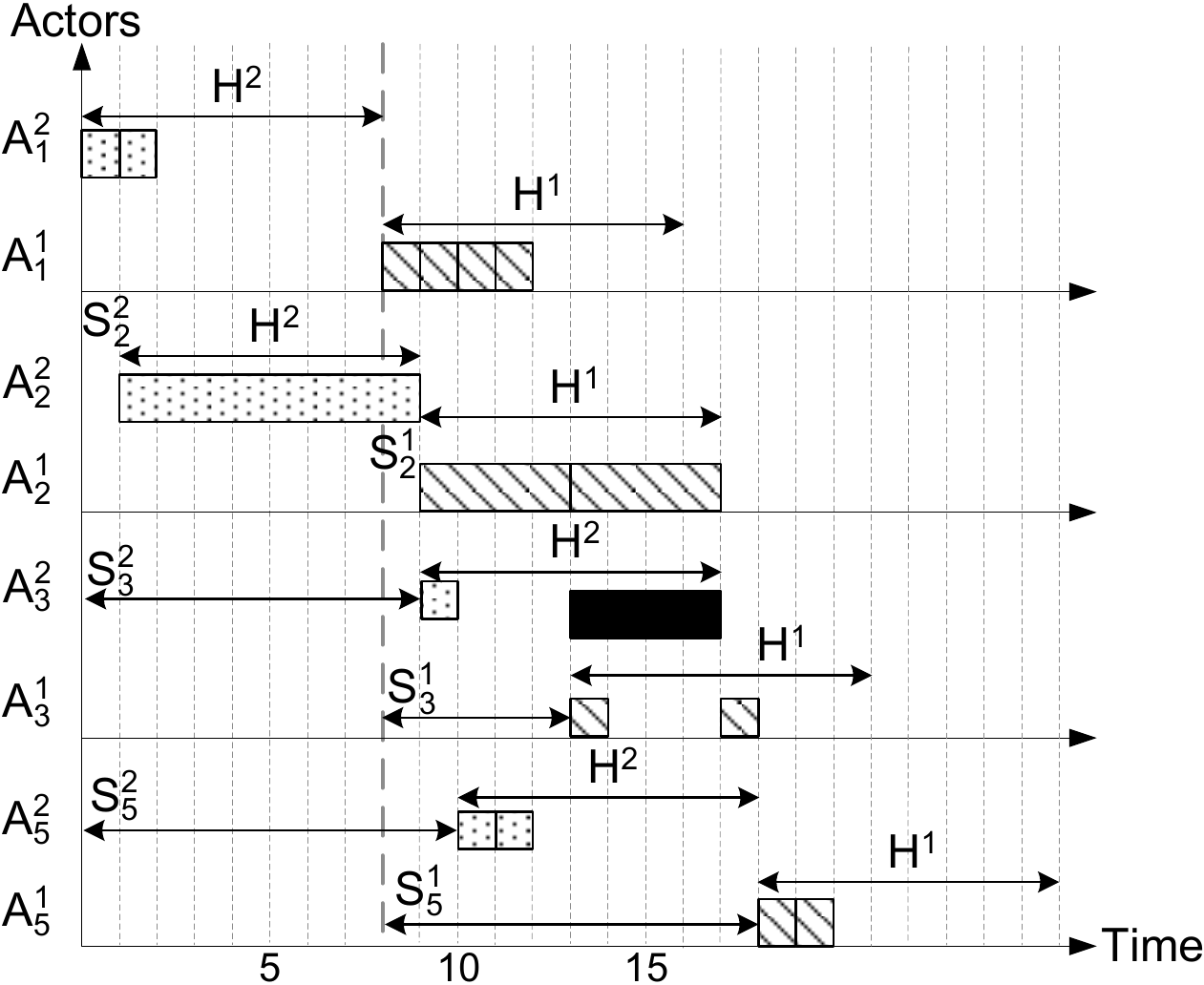}
	\caption{An illustration of the Maximum-Overlap Offset~(MOO) calculation.}
	\label{fig:moo}
\end{figure}
Basically, we first assume that the new mode $\textit{SI}^l$ starts immediately after the source actor $A_\text{src}^o$ of the old mode $\textit{SI}^o$ completes its last iteration $\textit{It}^o$.
All actors $A_i^l$ of the new mode execute according to the earliest starting times $S_i^l$ and iteration period $H^l$ in the steady-state.
Under this assumption, if the execution of the new mode overlaps with the execution of the old mode in terms of iteration periods $H^o$ and $H^l$,
we then need to offset the starting time of the new mode by the maximum overlap among all actors.
In this way, the execution of the new mode will have the same iteration latency as that of the new mode in the steady-state, i.e., no interference between the execution of both old and new modes.

Consider \textit{MCR1} at time $t_\text{MCR1}$ shown in \FIG~\ref{fig:st_transition}.
Obtaining MOO $x$ is illustrated in  \FIG~\ref{fig:moo}.
We first assume that the new mode $\textit{SI}^1$ starts at the time when the source actor $A_1^2$ completes the last iteration at time 8 (see bold, dashed line in \FIG~\ref{fig:moo}).
Actors $A_i^1$ in the new mode start as if they executed in the steady-state of mode $\textit{SI}^1$.
Then, we can see that, for actor $A_3$, the execution of $A_3^1$ in the new mode $\textit{SI}^1$ according to $S^1_3$ in \FIG~\ref{fig:static_ex1} overlaps 4 time units (solid bar in \FIG~\ref{fig:moo}) with the execution of $A_3^2$ in the old mode $\textit{SI}^2$ in terms of iteration periods $H^2$ and $H^1$.
This is also the maximum overlap between the execution of actors in modes $\textit{SI}^2$ and $\textit{SI}^1$.
According to \DEF~\ref{def:moo}, $x$ can be obtained through the following equations: \begin{align*}
S_1^2 - S_1^1 =  0 - 0 = 0,~~
S_2^2 - S_2^1 =  1 - 1 = 0,~~\\
S_3^2 - S_3^1 =  9 - 5 = 4,~~
S_5^2 - S_5^1 =  10 -10 = 0.
\end{align*}
Therefore, it results in an offset $x = \max(0,0,4,0)=4$ to the start of mode $\textit{SI}^1$ and is shown in \FIG~\ref{fig:moo_transition}.
\begin{figure}[!t!b]
	\centering
	\includegraphics[width=1\columnwidth]{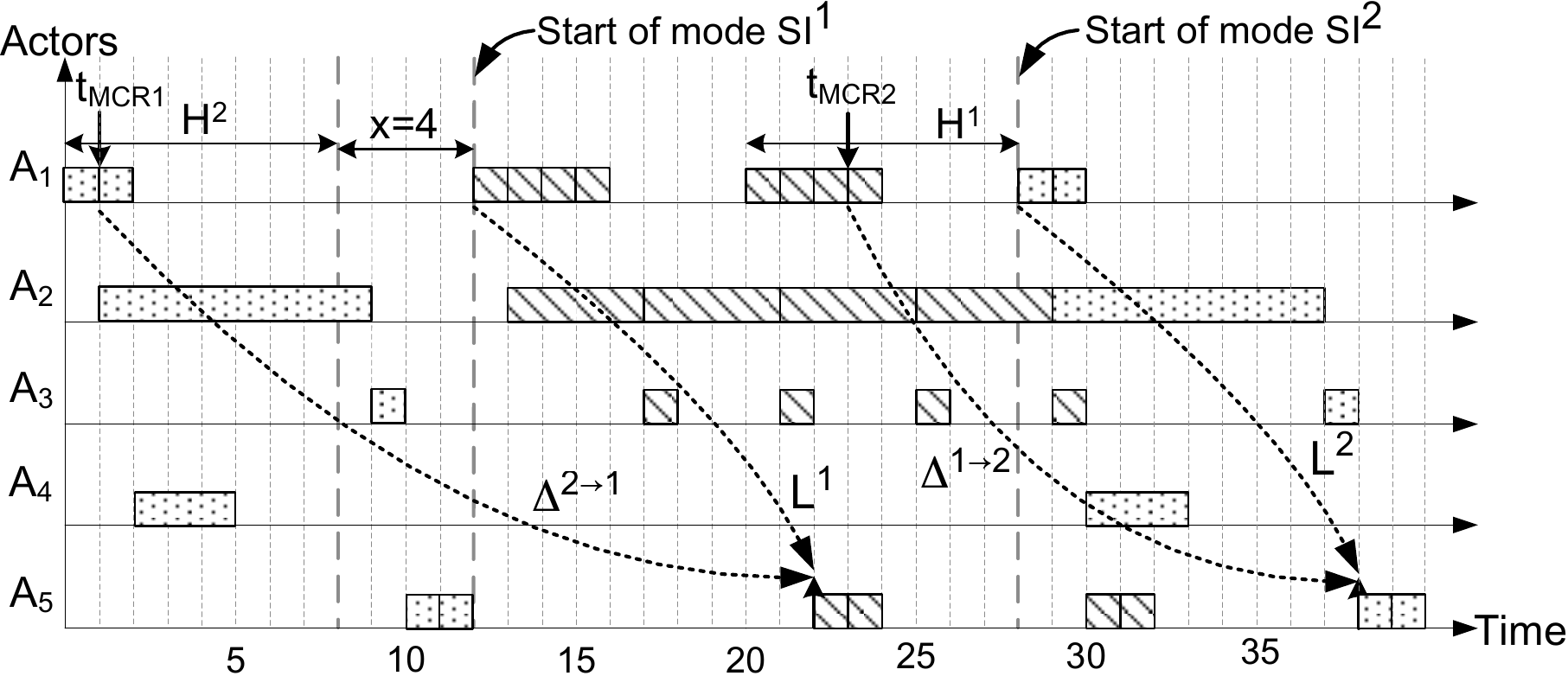}
	\caption{The execution of $G_1$ with two mode transitions under Maximum-Overlap Offset~(MOO) protocol.}
	\label{fig:moo_transition}
\end{figure}
The starting time of the new mode $\textit{SI}^1$, namely the source actor $A_1^1$, must be first delayed to the time when $A_2^1$ completes the iteration period $H^2$ in the last iteration, namely time $8$ shown as the first bold dashed line in \FIG~\ref{fig:moo_transition}.
In addition, the MOO $x = 4$ must be further added to the starting time of $A_1^1$~(the second bold dashed line in \FIG~\ref{fig:moo_transition}).
\FIG~\ref{fig:moo_transition} also shows another transition from mode $\textit{SI}^1$ to $\textit{SI}^2$ with a {MCR} occurred at time $t_\text{MCR2} = 23$.
The starting time of the source actor $A_1^2$ in the new mode $\textit{SI}^2$ must be first delayed to the time $28$~(the third bold dashed line in  \FIG~\ref{fig:moo_transition}), namely the time when $A^1_1$ completes the last iteration in the old mode $\textit{SI}^1$.
To calculate the MOO $x$ for this transition, the following equations hold: \begin{align*}
S_1^1 - S_1^2 =  0 - 0 = 0,~~
S_2^1 - S_2^2 =  1 - 1 = 0,~~\\
S_3^1 - S_3^2 =  5 - 9 = -4,~~
S_5^1 - S_5^2 =  10 -10 = 0.
\end{align*}
Thus, the equations above result in $x = \max(0,0,-4,0)=0$.
For this transition, the new mode $\textit{SI}^2$ starts at time $28$ as shown in  \FIG~\ref{fig:moo_transition}.

The MOO protocol offers several advantages over the ST protocol. 
Essentially, the MOO protocol retains the iteration latency of the \SADF{} graph in the new mode the same as the initial value, thereby avoiding the interference between the old and new modes.
For instance, after \textit{MCR1} and \textit{MCR2} in  \FIG~\ref{fig:moo_transition}, mode $\textit{SI}^1$ and $\textit{SI}^2$ still have the initial iteration latency $L^1= 10 $ and $ L^2=10$ as shown in \FIG~\ref{fig:static_ex}.
Therefore, efficiently computing the starting time of \SADF{} actors in the new mode becomes feasible and it plays an important role in deriving a hard-real time schedule for the \SADF{} actors.
As a result, analysis of the worst-case transition delay is much simpler~(see Theorem \ref{theo:min_max_delay}) than that of the ST protocol, because the transition
delay does not depend on the order of the transitions that occurred previously.

Concerning the transition delay, it may be the case that the MOO protocol results in initially longer transition delay than the ST protocol does due to the offset given in \DEF~\ref{def:moo}.
For \textit{MCR1} occurred at time $t_\text{MCR1}$, the transition delay of the MOO protocol is $\Delta^{2 \rightarrow 1} = 22 - 1 = 21$ as shown in \FIG~\ref{fig:moo_transition},
whereas the transition delay of the ST protocol is equal to $\Delta^{2 \rightarrow 1} = 18 - 1 = 17$ as shown in \FIG~\ref{fig:st_transition}.
On the other hand,
let us consider the same transition request \textit{MCR2} occurred at time $t_\text{MCR2} = 23$ shown in \FIG~\ref{fig:st_transition} and \FIG~\ref{fig:moo_transition}.
For \textit{MCR2}, the ST protocol results in transition delay $\Delta^{1\rightarrow 2} = 42 - 23 = 19$ as shown in \FIG~\ref{fig:st_transition}.
In contrast, the transition delay for the MOO protocol is $\Delta^{1 \rightarrow 2} = 38 - 23 = 16$ as shown in \FIG~\ref{fig:moo_transition}.
The MOO protocol could provide shorter transition delay than the ST protocol, thereby faster responsiveness to a mode transition. \section{Hard Real-Time Analysis and Scheduling of \SADF}\label{sec:hrt_sacsdf}
Based on the proposed MOO protocol for mode transitions, in this section, we propose a hard real-time analysis and scheduling framework for \SADF{}.
More specifically, we propose an analysis technique for mode transitions in \SADF{} to reason about transition delays, such that timing constraints can be guaranteed.
The hard real-time scheduling framework for \SADF{} graphs is an extension of the SPS~\cite{Bamakhrama2011_emsoft} framework initially developed for CSDF graphs.

As explained in \Cref{SPS}, the key concept of the SPS framework is to derive a periodic taskset representation for a CSDF graph.
Since the steady-state of a mode can be considered as a CSDF graph according to Definitions~\ref{def:sacsdf_sn} and \ref{def:sacsdf_steady}, it is thus straightforward to represent the steady-state of a \SADF{} graph as a periodic taskset and schedule the resulting taskset using any well-known hard real-time scheduling algorithm.
Using the SPS framework, we can derive the two main parameters for each \SADF{} actor in mode $\textit{SI}^k$, namely the period ($T_i^k$ in \Cref{period}) and the earliest starting time ($S_i^k$ in \Cref{start_time}).
Under SPS, the iteration period in mode $\textit{SI}^k$ is obtained as $H^k = q_i^k T_i^k,\ \exists A_i^k \in \mathcal{A}$.
Below, we focus on determining the earliest starting time of each actor in the new mode upon a transition.
From the earliest starting time, we can reason about the transition delay to quantify the responsiveness of a transition.

Upon a MCR, a \SADF{} graph can safely switch to the new mode if all of its actors have completed their last iteration in the old mode upon synchronous protocol.
In this case, the firings of \SADF{} actors in the new mode do not overlap with the firings of actors in the old mode. This is called synchronous protocol~\cite{Real2004} in real-time systems with mode change.
One of its advantages is the simplicity, i.e., the synchronous protocol does not require any schedulability test at both compile-time and run-time.
However, other protocols lead to earlier starting times than the synchronous protocol. Therefore, the synchronous protocol sets an upper bound on the earliest starting time for each \SADF{} actor in the new mode.
\begin{lemma}
	For a \SADF{} graph $G$ under SPS and a MCR from mode $\textit{SI}^o$ to $\textit{SI}^l$ at time $t_\text{MCR}$,
	the earliest starting time of actor $A_i^l$, $\hat{\sigma}_i^{o \rightarrow l}$, is upper bounded by
	\begin{equation}\label{eq:start_upper}
	\hat{\sigma}_i^{o \rightarrow l} = {F}_\text{src}^o + S_\text{snk}^o + S_i^l,
	\end{equation}
	where ${F}_\text{src}^o$ indicates the time when the source actor $A_\text{src}^o$ completes its last iteration $\textit{It}^o$ of the old mode $\textit{SI}^o$ and is given by
	\begin{equation}\label{eq:complete_src}
	{F}_\text{src}^o =  t_S^o + \bigg\lceil \frac{t_\textit{MCR} - t_S^o}{H^o} \bigg\rceil H^o.
	\end{equation}
	$t_S^o$ is the starting time of mode $\textit{SI}^o$ and $H^o$ is the iteration period of mode $\textit{SI}^o$.
\end{lemma}
\begin{proof}
	As explained previously for a transition from mode $\textit{SI}^o$ to $\textit{SI}^l$, the upper bound of the earliest starting time for each actor $A_i^l$ is computed in such a way that no firings of actors $A_i^o$ and $A_i^l$ occur simultaneously.
	This means, the start of an actor $A_i^l$ must be later than all actors $A_i^o$ have completed the last iteration $\textit{It}^o$
	of the old mode $\textit{SI}^o$.
	Given that mode $\textit{SI}^o$ starts at time $t_S^o$,
	the completion time of all actors $A_i^o$ in the last iteration $\textit{It}^o$ can be thus computed as
	\begin{equation}\label{eq:completion_snk}
		F_\text{snk}^o = t_S^o + \bigg\lfloor \frac{t_\textit{MCR} - t_S^o}{H^o} \bigg\rfloor H^o + S_\text{snk}^o + H^o.
	\end{equation}
	where $F_\text{snk}^o$ is the time when the old mode $\textit{SI}^o$ completes the last iteration $\textit{It}^o$.
	It is assumed that the sink actor $A_\text{snk}^o$ is the last actor to complete
	the iteration, i.e., $\forall A_i^o \in \mathcal{A}, S_i^o \le S_\text{snk}^o$.
	Given \Cref{eq:complete_src}, \Cref{eq:completion_snk} can be rewritten as
	\begin{equation*}
		F_\text{snk}^o = t_S^o + \bigg\lceil \frac{t_\textit{MCR} - t_S^o}{H^o} \bigg\rceil H^o + S^o_\text{snk} = {F}_\text{src}^o + S_\text{snk}^o.
	\end{equation*}
	Now, starting the source actor $A_\text{src}^l$ at any time later than $F_\text{snk}^o$ is valid without introducing simultaneous execution of actors $A_i^o$ and $A_i^l$.
	Therefore, the earliest starting time of source actor $A_\text{src}^l$ is $\hat{\sigma}_\text{src}^{o \rightarrow l} = F_\text{snk}^o$.
	For any actor $A_i^l \in \mathcal{A} \setminus A_\text{src}^l$, its earliest starting times must satisfy \Cref{start_time} imposed by the SPS framework.
	That is, the earliest starting time $\hat{\sigma}_i^{o \rightarrow l}$ of actor $A_i^l$ can be obtained by adding $S_i^l$ to $ \hat{\sigma}_\text{src}^{o \rightarrow l}$.
\end{proof}

\begin{figure}[!t!b]
	\centering
	\includegraphics[width=1\columnwidth]{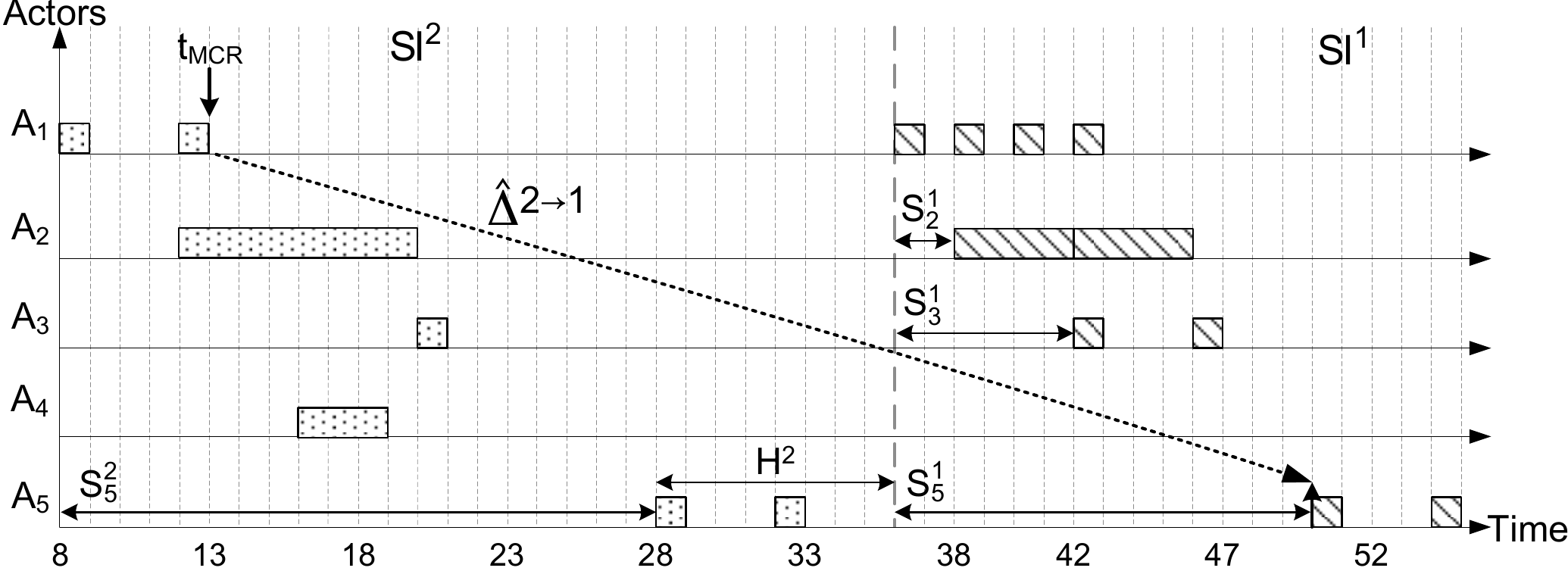}
	\caption{Upper bounds of earliest starting times for transition from mode $\textit{SI}^2$ to $\textit{SI}^1$.
	}
	\label{fig:start_upper}
\end{figure}
Let us consider the actor parameters given in \TAB~\ref{tab:sacsdf_param_ex2} for $G_1$ in \FIG~\ref{fig:sacsdf_ex}.
The third row shows the WCET for each actor in modes $\textit{SI}^1$ and $\textit{SI}^2$.
Based on WCETs, the period (fourth row in \TAB~\ref{tab:sacsdf_param_ex2}) and the earliest starting time~(fifth row in \TAB~\ref{tab:sacsdf_param_ex2}) for each actor in the steady-state of both modes are obtained according to \Cref{period} and \Cref{start_time}, respectively.
Given $\vec{q}^2$ in \Cref{eq:sacsdf_q2}, we can also compute iteration period $H^2 = q^2_1 T^2_1 = 2 \times 4 = 8$.
Now consider the mode transition from mode $\textit{SI}^2$ to $\textit{SI}^1$ shown in \FIG~\ref{fig:start_upper}.
Assume that the MCR occurs at time $t_\text{MCR} = 13$ and mode $\textit{SI}^2$ starts at time $t_S^2 = 8$.
The completion time of the last iteration $\textit{It}^2$
is equal to the completion time of the sink actor $A_5^2$ computed as
\begin{equation*}
{F}_\text{snk}^2 = t_S^2 + \bigg \lceil \frac{t_\text{MCR} - t_S^2}{H^2} \bigg\rceil H^2 + S_5^2 = 8+ \bigg \lceil \frac{13 - 8}{8} \bigg\rceil  8 + 20 = 36.
\end{equation*}
In \FIG~\ref{fig:start_upper}, ${F}_\text{snk}^{2}$ corresponds to the earliest starting time of the source actor $A_1^1$ (bold dashed line).
Finally, we can compute the earliest starting time for each actor in the new mode $\textit{SI}^1$ by adding $S_i^1$.
Considering for instance the sink actor $A_5^1$ in the new mode with $S_5^1 = 14$, the upper bound of its earliest starting time can be obtained as
\begin{equation*}
\hat{\sigma}_5^{2 \rightarrow 1} = F^2_\text{src} + S_5^2 + S^1_5 = F^2_\text{snk} + S^1_5 = 36 + 14 = 50.
\end{equation*}
We can thus compute the transition delay (\textit{cf}. \DEF~\ref{def:trans_delay}) as
\begin{equation*}
\hat{\Delta}^{2\rightarrow 1} = \hat{\sigma}^{2\rightarrow 1}_5 - t_\text{MCR} = 50 - 13 = 37.
\end{equation*}

\begin{figure}[!t!b]
	\centering
	\includegraphics[width=1\columnwidth]{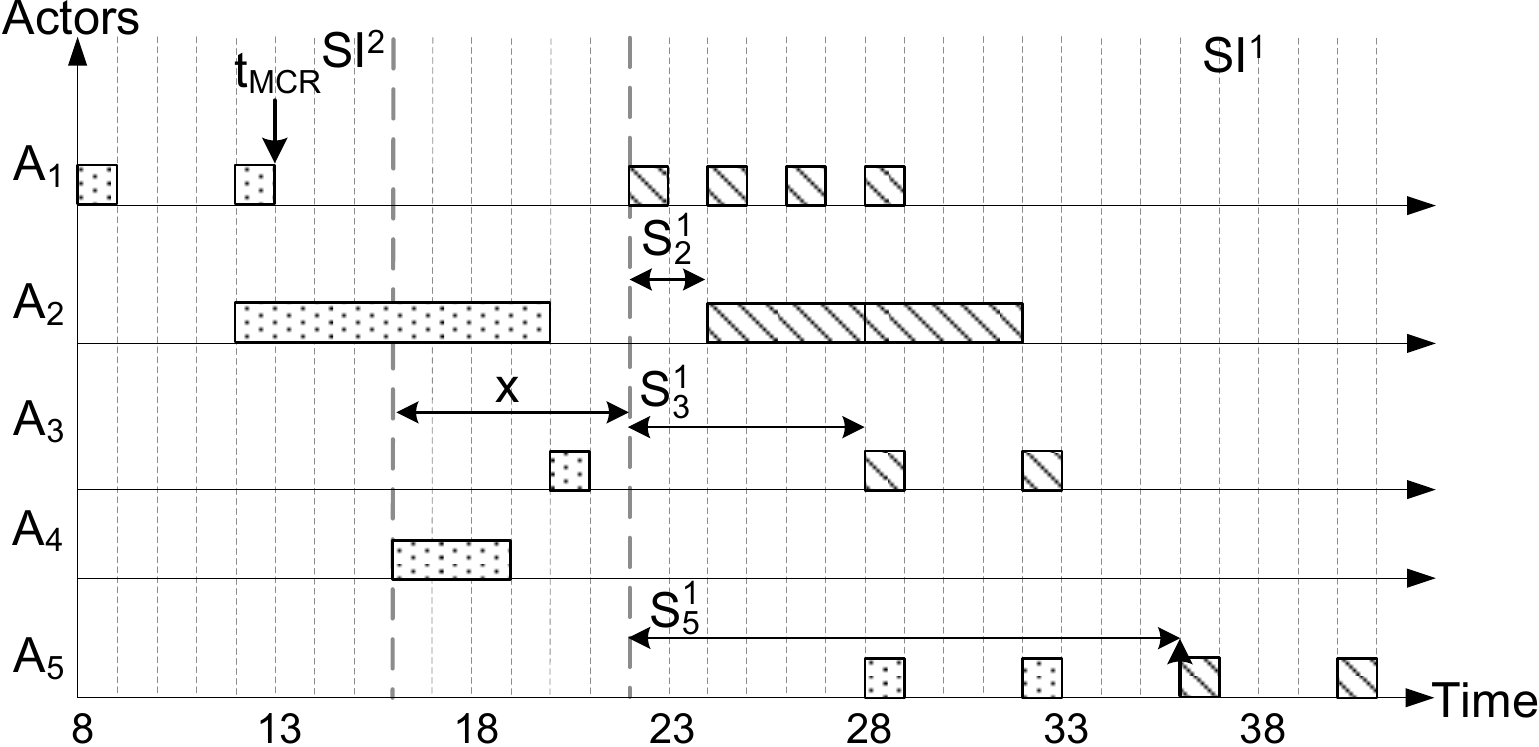}
	\caption{Earliest starting times for transition from mode $\textit{SI}^2$ to $\textit{SI}^1$ with the MOO protocol.
	}
	\label{fig:start_moo}
\end{figure}
Although the upper bound of the earliest starting times is easy to obtain for \SADF{} actors in the new mode, it does not provide a responsive mode transition.
Therefore, here we aim at deriving a lower bound of the earliest starting times with the proposed MOO protocol.
\begin{lemma}\label{lm:start_lower}
	For a \SADF{} graph under SPS and a MCR from mode $\textit{SI}^o$ to $\textit{SI}^l$ at time $t_\text{MCR}$,
	the earliest starting time of actor $A_i^l$ using the MOO protocol is lower bounded by $\check{\sigma}_i^{o \rightarrow l}$ given as
	\begin{equation}\label{eq:start_lower}
		\check{\sigma}_i^{o \rightarrow l} =  {F}_\text{src}^o + x + S_i^l,
	\end{equation}
	where ${F}^o_\text{src}$ is given in \Cref{eq:complete_src}
	and $x$ is given in \Cref{eq:moo_x}.
\end{lemma}
\begin{proof}
	Under the MOO protocol, the start of actor $A_i^l$ must be later than the time when $A_i^o$, if any, completes
	its last iteration in the old mode $\textit{SI}^o$.
	We assume that the source actor $A_\text{src}^l$ is the first actor to start in the new mode $\textit{SI}^l$, i.e.,
	$\forall A_i^l \in \mathcal{A}, S_i^l \ge S_\text{src}^l$.
	Thus, the starting time of the source actor $A_\text{src}^l$ is at least equal to the completion time of the last iteration of $A_\text{src}^o$, denoted by $F_\text{src}^o$.
	Given $F_\text{src}^o$ in \Cref{eq:complete_src}, it thus holds $\check{\sigma}_\text{src}^{o\rightarrow l} \ge F_\text{src}^o$.
	Then, the offset $x$ because of the MOO protocol given in \Cref{eq:moo_x} must be taken into account.
	Consequently, the earliest starting time of $A_\text{src}^l$ is lower bounded by $\check{\sigma}_\text{src}^{o\rightarrow l} = F_\text{src}^o + x$.
	For any actor $A_i^l \in \mathcal{A} \setminus A_\text{src}^l$, its earliest starting times must satisfy \Cref{start_time} imposed by the SPS framework.
	Hence, the earliest starting time $\check{\sigma}_i^{o \rightarrow l}$ of actor $A_i^l$ can be obtained by adding $S_i^l$ to $ \check{\sigma}_\text{src}^{o \rightarrow l}$.
\end{proof}

Let us consider again the transition from mode $\textit{SI}^2$ to $\textit{SI}^1$.
With the MOO protocol, the mode transition is illustrated in \FIG~\ref{fig:start_moo}.
Upon the MCR at time $t_\text{MCR} = 13$ and $t_S^2 = 8$, source actor $A_1^2$ completes its last iteration $\textit{It}^2$ in the old mode $\textit{SI}^2$ at the time (\textit{cf}. \Cref{eq:complete_src}) given as
\begin{equation*}
	{F}_\text{src}^2 = {F}_1^2 = t_S^2 + \bigg \lceil \frac{t_\text{MCR} - t_S^2}{H^2} \bigg \rceil  H^2 = 8 + \bigg\lceil \frac{13 - 8}{8} \bigg\rceil  8 = 16.
\end{equation*}
This is the earliest possible time at which mode transition is allowed.
For MOO, $x$ can be computed according to \Cref{eq:moo_x}.
Therefore, the following equations hold:
\begin{align*}
S_1^2 - S_1^1 = 0 - 0 = 0,~
S_2^2 - S_2^1 = 4 - 2 = 2,~\\
S_3^2 - S_3^1 = 12 - 6 = 6,~
S_5^2 - S_5^1 = 20 - 14 = 6.
\end{align*}
It thus yields $x= \max(0,2,6,6) = 6$, i.e., an offset $x=6$ is added to $F_\text{src}^2$.
It can be seen in \FIG~\ref{fig:start_moo} that the source actor $A_1^1$ starts at time $F_\text{src}^2+x = 16 + 6 = 22$.
Finally, the earliest starting times of actors in mode $\textit{SI}^1$ can be determined by adding $S_i^1$.
Considering for instance $A_5^1$ in the new mode, the lower bound of its earliest starting time can be obtained as:
\begin{equation*}
	\check{\sigma}_5^{2 \rightarrow 1} = {F}_\text{src}^2 + x + S_5^1 = 16 + 6 + 14 = 36.
\end{equation*}
Now, the transition delay (\textit{cf}. \DEF~\ref{def:trans_delay}) can be obtained as
\begin{equation*}
	\check{\Delta}^{2\rightarrow 1} = \check{\sigma}_5^{2\rightarrow 1} - t_\text{MCR} = 36 - 13 = 23.
\end{equation*}

\subsection{Scheduling Analysis under a Fixed Allocation of Actors}
During a mode transition of a \SADF{} graph according to the MOO protocol, actors execute simultaneously in the old and new modes.
The derived starting time in Lemma \ref{lm:start_lower} for each actor is only the lower bound because the allocation of actors on PEs is not taken into account yet.
That means, the derived starting times according to Lemma \ref{lm:start_lower} can be only achieved during mode transitions when each actor is allocated to a separate PE.
In a practical system where multiple actors are allocated to the same PE,
the PE may be potentially overloaded during mode transitions.
To avoid overloading of PEs, the earliest starting times of actors may be further delayed.
\begin{lemma}\label{lm:overload_offset}
	For a \SADF{} graph under SPS, a MCR from mode $\textit{SI}^o$ to $\textit{SI}^l$, and a $m$-partition of all actors $\Psi = \{\Psi_1,\dots,\Psi_m\}$, where $m$ is the number of PEs,
	the earliest starting time of an actor $A_i^{l}$ without overloading the underlying PE is given by
	\begin{equation}\label{eq:sacsdf_start}
		\sigma_i^{o\rightarrow l} = {F}_\text{src}^o + \delta^{o \rightarrow l} + S_i^l,
	\end{equation}
	where ${F}_\text{src}^o$ is computed by \Cref{eq:complete_src} and $\delta^{o \rightarrow l}$ is obtained as
	\begin{equation}\label{eq:overload_offset}
		\delta^{o\rightarrow l} = \min_{t \in [x, S^o_\text{snk}]} \{t: U_j(k) \leq \textit{UB},
		~{\forall k \in [t, S_\text{snk}^o]} \land \forall \Psi_j \in \Psi\}.
	\end{equation}
	\textit{UB} denotes the utilization bound of the scheduling algorithm used to schedule actors on each PE.
	$\Psi_j$ contains the set of actors allocated to PE$_j$.
	$U_j(k)$ is the total utilization of PE$_j$ at time $k$ demanded by both mode $\textit{SI}^o$ and $\textit{SI}^l$ actors, and is given by
	\begin{equation}\label{eq:total_u}
		U_j(k) = \underbrace{\sum_{A_{d}^o \in \Psi_j} \Big( u_d^o -h( k - S_d^o ) \cdot u_d^o \Big)}_{U_j^o(k)} + \underbrace{\sum_{A_{d}^l \in \Psi_j} \Big( h( k - S_d^l -t ) \cdot u_d^l \Big)}_{U_j^l(k)},
	\end{equation}
	$A_{d}^o \in \Psi_j$ is an actor active in the old mode $\textit{SI}^o$ and allocated to PE$_j$.
	$A_{d}^l \in \Psi_j$ is an actor active in the new mode $\textit{SI}^l$ and allocated to PE$_j$.
	$h(t)$ is the Heaviside step function.
\end{lemma}

\begin{proof}
	Lemma \ref{lm:start_lower}  shows the lower bound of the earliest starting time for actor $A_i^l$ in the new mode $\textit{SI}^l$.
	However, starting $A_i^l$ at time $\check{\sigma}_i^{o \rightarrow l}$ may overload PE$_j$, i.e.,
	the resulting total utilization of PE$_j$, denoted by $U_j(\check{\sigma}_i^{o \rightarrow l})$, exceeds $\textit{UB}$.
	Therefore, in this case, the earliest starting time $\sigma_i^{o \rightarrow l}$ must be delayed by $\delta^{o\rightarrow l}$ such that $U_j({\sigma}_i^{o\rightarrow l}) \le \textit{UB}$ holds.
	From \Cref{eq:sacsdf_start} and \Cref{eq:start_lower}, we can see that $\delta^{o\rightarrow l}$ is lower bounded by $x$ which corresponds to the MOO protocol.
	In addition, $\delta^{o\rightarrow l}$ is upper bounded by $S_\text{snk}^o$ if we consider \Cref{eq:sacsdf_start} and \Cref{eq:start_upper}.

	$\delta^{o\rightarrow l}$ of interest is the minimum time $t$ in the bounded interval $[x, S_\text{snk}^o]$ that satisfies two conditions.
	
	\underline{Condition 1:} For each PE$_j$, the total utilization cannot exceed \textit{UB} at time $t$, i.e., $U_j(t) \le \textit{UB}$.
	The total utilization $U_j(t)$ in \Cref{eq:total_u} consists of two parts, namely $U_j^o(t)$ and $U_j^l(t)$.
	$U_j^o(t)$ denotes the PE capacity occupied by the actors in mode $\textit{SI}^o$ that are not completed yet.
	Additional PE capacity $U_j^l(t)$ is demanded by the already released actors in the new mode $\textit{SI}^l$.
	
	\underline{Condition 2:} We need to check all time instants $k > t$ in the interval $[t, S_\text{snk}^o]$, such that $U_j(k) \le \textit{UB}$, to guarantee that each PE$_j$ is not overloaded during the mode transition.

\end{proof}
\begin{figure}[!t!b]
	\centering
	\includegraphics[width=.65\columnwidth]{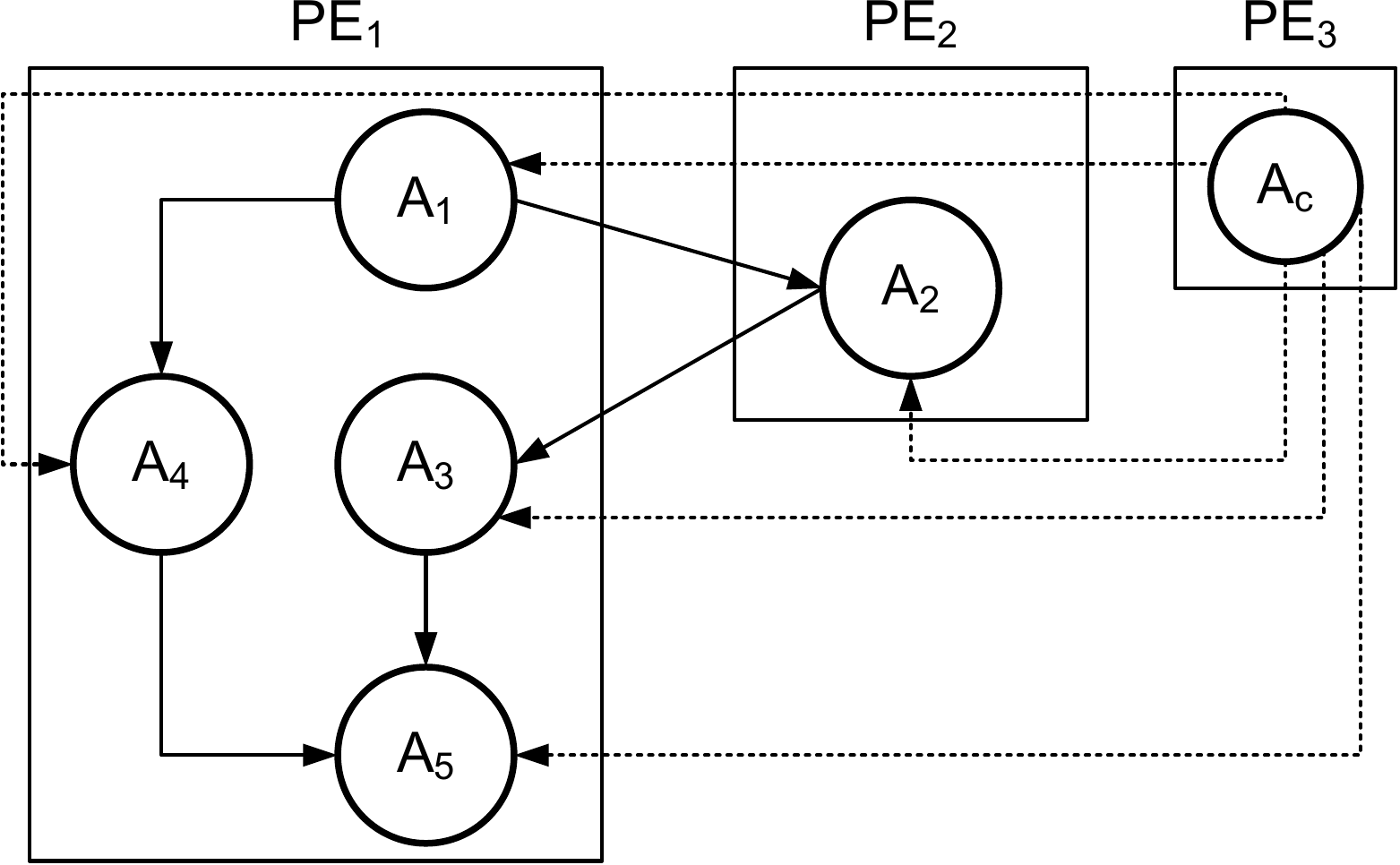}
	\caption{Allocation of all \SADF{} actors in \FIG~\ref{fig:sacsdf_ex} to 3 PEs.}
	\label{fig:sacsdf_alloc}
\end{figure}

\FIG~\ref{fig:sacsdf_alloc} shows all actors of $G_1$ in \FIG~\ref{fig:sacsdf_ex} allocated to 3 PEs and let us assume that the actors allocated to each PE are scheduled using the EDF scheduling algorithm \cite{Liu1973_hrt}. The utilization bound of EDF is given in \cite{Liu1973_hrt} as $\textit{UB} = 1$. Given this allocation and the transition from mode $\textit{SI}^2$ to $\textit{SI}^1$ shown in \FIG~\ref{fig:start_moo}, the lower bound of the earliest starting time $\check{\sigma}_1^{2 \rightarrow 1} = 22$ for actor $A_1^1$ cannot be achieved. At time $22$, only actor $A_1^2$ has completed the last iteration $\textit{It}^2$ on PE$_1$. Starting the new mode $\textit{SI}^1$ at time $22$ corresponds to $\delta^{2 \rightarrow 1} = x = 6$. The total utilization of PE$_1$ demanded by the actors in the old mode $\textit{SI}^2$ at time $22$, i.e., $U_1^2(6)$, can be computed as follows:

{\footnotesize
\begin{align*}
	U_{1}^2 (6) &=\sum_{A_d^2 \in \Psi_1} u_d^2 - h(6 - S_d^2) \cdot u_d^2,~d \in \{1,3,4,5\} \\
	&= u_1^2 - h(6) \cdot u_1^2 + u_3^2 - h(-6) \cdot u_3^2 + u_4^2 - h(-2) \cdot u_4^2 + u_5^2 - h (-14) \cdot u_5^2 \\
	&= 0 + u_3^2 + u_4^2 + u_5^2 = \frac{1}{8} + \frac{3}{8} + \frac{1}{4} = \frac{3}{4}.
\end{align*}}Enabling $A_1^1$ in the new mode $\textit{SI}^1$ at time 22 would yield
\begin{equation*}
	U_1(6) = U_1^2(6) + u_1^1 = \frac{3}{4} + \frac{1}{2} > \textit{UB} = 1,
\end{equation*}
thereby leading to being unschedulable on PE$_1$. In this case, the earliest starting times of all actors in mode $\textit{SI}^1$ must be delayed by $\delta^{2\rightarrow 1} = 8$ to time 24 as shown in \FIG~\ref{fig:start_alloc}.
At time 24, the total utilization demanded by mode $\textit{SI}^2$ actors is
{\footnotesize
\begin{align*}
	U_{1}^2 (8) &= \sum_{A_d^2 \in \Psi_1} u_d^2 - h(8 - S_d^2) \cdot u_d^2,~d \in \{1,3,4,5\} \\
	&= u_1^2 - h(8) \cdot u_1^2 + u_3^2 - h(-4) \cdot u_3^2 + u_4^2 - h(0) \cdot u_4^2 + u_5^2 - h (-12) \cdot u_5^2 \\
	&= 0 + u_3^2 + 0 + u_5^2 = \frac{1}{8} + \frac{1}{4} = \frac{3}{8}. 
\end{align*}}Now, enabling $A_1^1$ in the new mode at time 24 results in the total utilization of PE$_1$ as
\begin{equation*}
	U_1(8) = U_1^2(8) + u_1^1 = \frac{3}{8} + \frac{1}{2} < 1.
\end{equation*}
Next, assuming that the new mode $\textit{SI}^1$ starts at time $24$, we need to check that the remaining actors in the new mode $\textit{SI}^1$, namely $A_3^1$ and $A_5^1$, can start with $S_3^1$ and $S_5^1$ respectively without overloading PE$_1$.
For instance, enabling $A_3^1$ at time 24 results in starting time $\sigma_3^{2\rightarrow 1} = 24 + S_3^1 = 24 + 6 = 30$.
At time 30, the total utilization of PE$_1$ can be obtained according to \EQ~(\ref{eq:total_u}) as follows:
{\footnotesize
\begin{align*}
	U_{1}^2 (8 + 6) &= \sum_{A_d^2 \in \Psi_1} u_d^2 - h(14 - S_d^2) \cdot u_d^2,~ d \in \{1,3,4,5\} \\
	&= u_1^2 - h(14) \cdot u_1^2 + u_3^2 - h(2) \cdot u_3^2 + u_4^2 - h(6) \cdot u_4^2 + u_5^2 - h (-6) \cdot u_5^2 \\
	&= 0 + 0 + 0 + u_5^2 = \frac{1}{4},\\
	U_{1}^1 (8 + 6) &= \sum_{A_d^1 \in \Psi_1} \Big( h( 14 - S_d^1 -8 ) \cdot u_d^1 \Big),~ d \in \{1,3,5\} \\
	&= h(6)u_1^1 + h(0)u_3^1 + h(-8)u_5^1= \frac{1}{2} + \frac{1}{4} = \frac{3}{4},\\
	U_1(8 + 6) &= U_{1}^2 (8+6) + U_{1}^1 (8 + 6) = 1 = \textit{UB}.
\end{align*}}Hence, actors $A_5^2$, $A_1^1$, and $A_3^1$ are schedulable on PE$_1$ using EDF.
Similarly, starting $A_5^1$ at time $\sigma_5^{2\rightarrow 1} = 24 + S_5^1 = 38$ still keeps the resulting set of actors schedulable on PE$_1$.

\begin{figure}[!t!b]
	\centering
	\includegraphics[width=1\columnwidth]{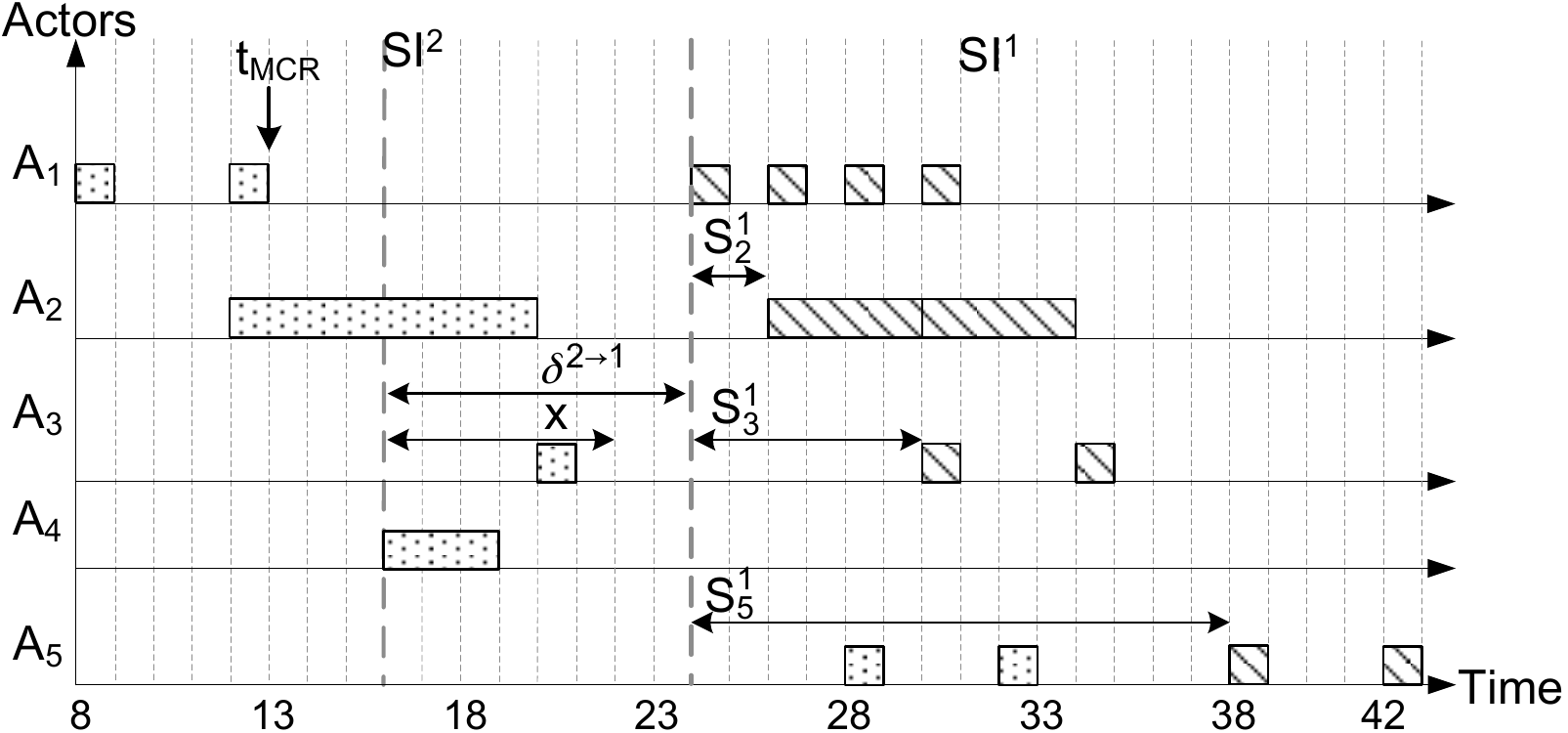}
	\caption{Earliest starting times for transition $\textit{SI}^2$ to $\textit{SI}^1$ on 3 PEs shown in \FIG~\ref{fig:sacsdf_alloc}.}
	\label{fig:start_alloc}
\end{figure}
Using Lemma \ref{lm:overload_offset}, we can quantify the maximum and minimum transition delays for any transition from mode $\textit{SI}^o$ to $\textit{SI}^l$.
\begin{figure*}[!t!b]
	\centering
	\begin{minipage}{.59\textwidth}
		\centering
		\includegraphics[width=1\columnwidth,height=0.32\columnwidth]{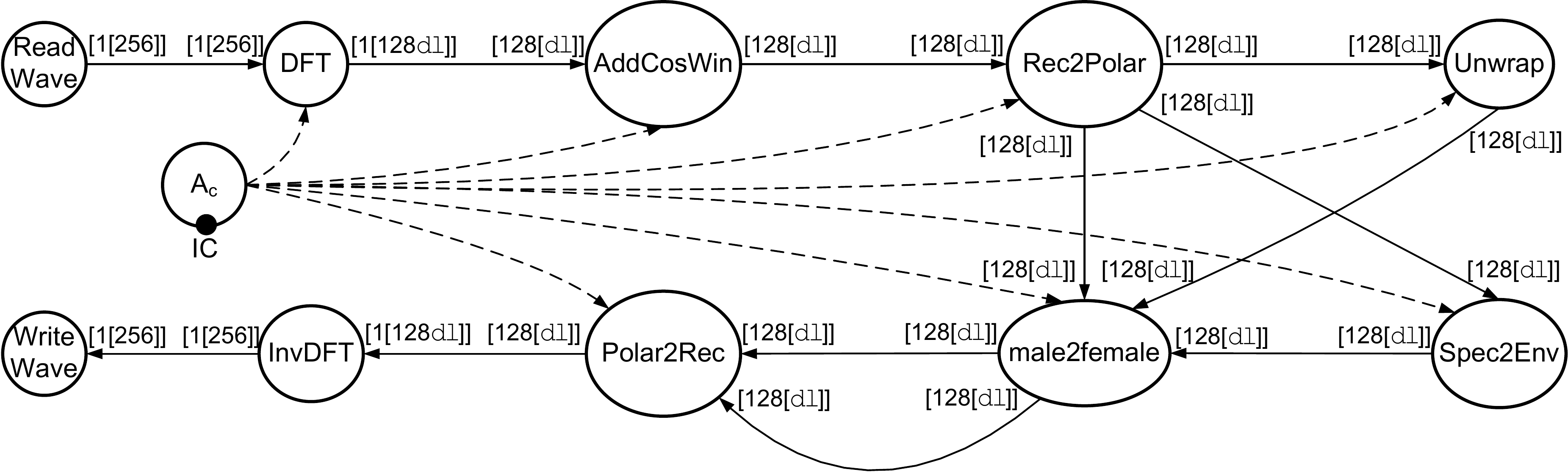}
		\caption{\SADF{} graph of Vocoder.}
		\label{fig:vocoder_sacsdf}
	\end{minipage}
	\hspace{2mm}
	\begin{minipage}{.35\textwidth}
		\centering
		\includegraphics[width=1\columnwidth,height=0.35\columnwidth]{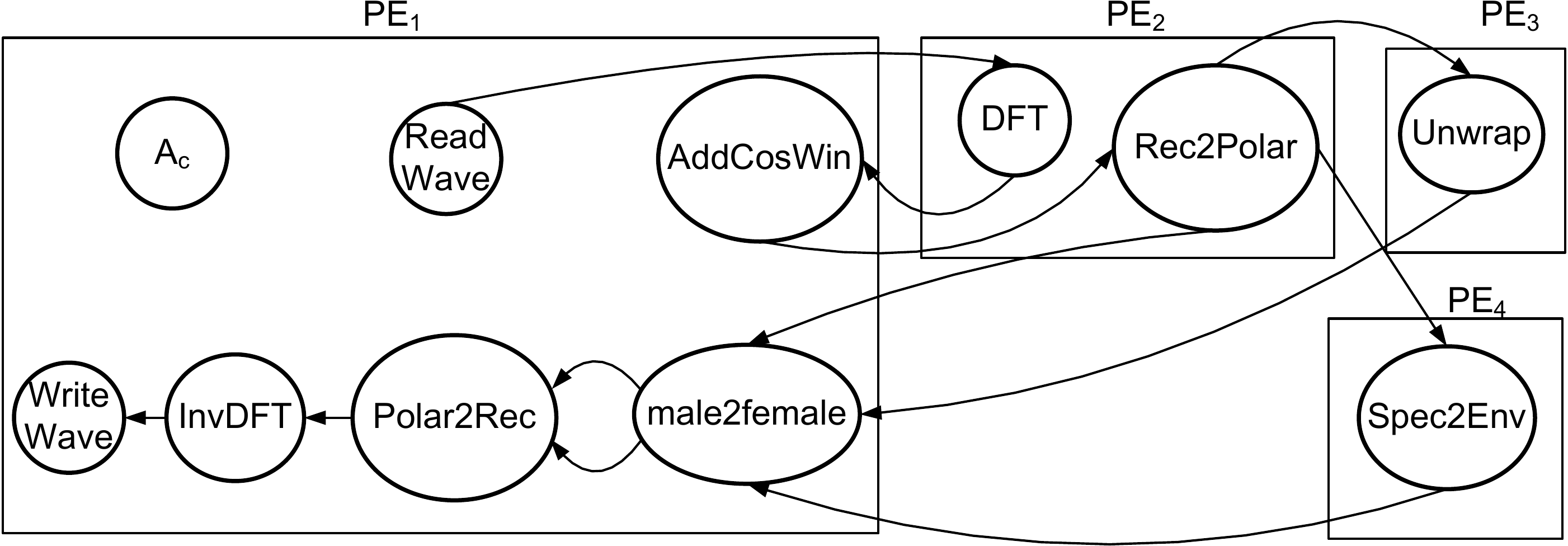}
		\caption{Allocation of dataflow actors of Vocoder to 4 PEs.
			The control edges are omitted to avoid cluttering.}
		\label{fig:vocoder_alloc}
	\end{minipage}
\end{figure*}
\begin{theorem}\label{theo:min_max_delay}
	For a \SADF{} graph under SPS, a fixed allocation of all \SADF{} actors $\Psi = \{\Psi_1,\dots,\Psi_m\}$ to $m$ PEs, and a MCR from mode $\textit{SI}^o$ to $\textit{SI}^l$,
	the minimum transition delay is given by
	\begin{equation}\label{eq:trans_best}
		\Delta^{o\rightarrow l}_\text{min} = \delta^{o\rightarrow l} + S_\text{snk}^l
	\end{equation}
	and the maximum transition delay is given by
	\begin{equation}\label{eq:trans_worst}
		\Delta^{o\rightarrow l}_\text{max} = \delta^{o\rightarrow l} + S_\text{snk}^l + H^o,
	\end{equation}
	where $\delta^{o\rightarrow l}$ is computed by Lemma \ref{lm:overload_offset}, $S^l_\text{snk}$ is the starting time of the sink actor in the new mode $\textit{SI}^l$, and $H^o$ is the iteration period of the old mode $\textit{SI}^o$.
\end{theorem}
\begin{proof}
	For a MCR from mode $\textit{SI}^o$ to $\textit{SI}^l$, the transition delay $\Delta^{o \rightarrow l}$ of a \SADF{} graph is given
	in \DEF~\ref{def:trans_delay} as $\Delta^{o \rightarrow l}= \sigma_\text{snk}^{o \rightarrow l} - t_\text{MCR}$, where
	the earliest starting time of the sink actor is calculated
	as $\sigma_\text{snk}^{o \rightarrow l} = F_\text{src}^o + \delta^{o \rightarrow l} + S^l_\text{snk}$ according to Lemma \ref{lm:overload_offset}.
	Therefore, $\Delta^{o \rightarrow l}$ can be rewritten as $\Delta^{o \rightarrow l}= F^o_\text{src} + \delta^{o \rightarrow l} + S_\text{snk}^l - t_\text{MCR}$.
	Essentially, $\Delta^{o \rightarrow l}$ is composed of three parts.
	In the first part, the MOO transition protocol together with a fixed allocation of the \SADF{} actors determine $\delta^{o \rightarrow l}$.
	The second part $S_\text{snk}^l$ results from the SPS framework.
	These two parts thus can be determined at compile-time.
	The third part $F^o_\text{src} - t_\text{MCR}$ depends on when the MCR occurs, namely at $t_\text{MCR}$, which can only be
	determined at run-time.
	In the following, we distinguish two cases for $t_\text{MCR}$:
	
	\underline{Case 1:} Assume that the MCR occurs at the end of an iteration of the source actor in the old mode $\textit{SI}^o$, i.e., $t_\text{MCR} = F^o_\text{src}$.
	Then, the source actor shall be only delayed by $\delta^{o \rightarrow l}$ to start in the new mode $\textit{SI}^l$ according to Lemma \ref{lm:overload_offset}, thereby guaranteeing the fastest possible start of the new mode $\textit{SI}^l$.
	As a consequence, it results in the minimum possible transition delay.
	Therefore, substituting $t_\text{MCR} = F^o_\text{src}$, we obtain
	\begin{equation*}
		\Delta^{o\rightarrow l}_\text{min} = F^o_\text{src} + \delta^{o\rightarrow l} + S_\text{snk}^l - F^o_\text{src} = \delta^{o\rightarrow l} + S_\text{snk}^l.
	\end{equation*}
	
	\underline{Case 2:} Assume that the MCR occurs at the beginning of an iteration of the source actor in the old mode $\textit{SI}^o$, i.e., $t_\text{MCR} = F^o_\text{src} - H^o$.
	Then, the source actor cannot start in the new mode before it completes the whole iteration in the old mode $\textit{SI}^o$ followed by the delay $\delta^{o \rightarrow l}$ according to Lemma \ref{lm:overload_offset}.
	Therefore, the maximum transition delay is computed as follows:
	\begin{equation*}
		\Delta^{o\rightarrow l}_\text{max} = F^o_\text{src} + \delta^{o\rightarrow l} + S_\text{snk}^l - (F^o_\text{src} - H^o) = \delta^{o\rightarrow l} + S_\text{snk}^l + H^o.
	\end{equation*}
\end{proof}

It can be seen from Theorem \ref{theo:min_max_delay} that the maximum and minimum transition delays solely depend on the allocation of \SADF{} actors and the old and new modes in question, irrespective of the previously occurred transitions.
The old and new modes determine $H^o$ and $S_\text{snk}^l$, respectively, while the allocation of \SADF{} actors determines the value of $\delta^{o \rightarrow l}$.
Here, the offset $x$ due to our MOO protocol is captured in $\delta^{o \rightarrow l}$ and can be considered as performance overhead if $x \neq 0$.
The other parts, namely $H^o$ and $S_\text{snk}^l$, in the maximum and minimum transition delays cannot be avoided as they will
be present in any transition protocol.
 \section{Case Studies}\label{sec:Case Study}
To evaluate our proposed \SADF{} MoC and MOO protocol, in this section, we present two case studies. In the first case study, we model a real-life adaptive streaming application, called Vocoder, with our \SADF{} MoC proposed in \SEC~\ref{sec:dyn_moc_def} and apply the hard real-time analysis proposed in \SEC~\ref{sec:hrt_sacsdf}. With this case study, we show that the \SADF{} MoC is capable of capturing different application modes and the transitions between them. Then, in the second case study, we model another real-life adaptive streaming application, called MP3decoder, with MADF and we focus on analyzing the transition delays and demonstrating the effectiveness of our MADF model armed with the proposed MOO transition protocol compared to the well-known FSM-SADF model \cite{Geilen2010_codes10} which also can capture modes/scenarios. In this case study, we adopt self-timed scheduling for both our MADF and FSM-SADF models in the steady-state. The major difference between these models in this case study is their transition protocol which is the MOO protocol in our MADF model and the self-timed protocol in FSM-SADF. Another example of the application of our MOO protocol can be found in \cite{niknam2017energy}.

\subsection{Case Study 1}
In this section, we consider a real-life adaptive application from the StreamIT benchmark suit~\cite{Gordon2006_asplos}, called Vocoder,
which implements a phase voice encoder and performs pitch transposition of recorded sounds from male to female.
We modeled Vocoder using a \SADF{} graph with 4 modes, which capture different workloads.
The \SADF{} graph of Vocoder is shown in \FIG~\ref{fig:vocoder_sacsdf}. Depending on the desired quality of audio encoding and various performance requirements, the resource manager as a middle-ware or OS-like component for the MPSoC may switch between four different modes of Vocoder at run-time.
The four modes $\mathcal{S} = \{\textit{SI}^8, \textit{SI}^{16}, \textit{SI}^{32}, \textit{SI}^{64}\}$ specify different lengths of the Discrete Fourier Transform~(DFT), denoted by $\SYM{dl} \in \{8, 16,32,64\}$.
Mode $\textit{SI}^8$ $(\SYM{dl} = 8)$ requires the least amount of computation at the cost of the worst voice encoding quality among all DFT lengths.
Mode $\textit{SI}^{64}$ $(\SYM{dl} = 64)$ produces the best quality of voice encoding among all modes, but is computationally intensive.
The other two modes $\textit{SI}^{16}$ and $\textit{SI}^{32}$ explore the trade-off between the quality of the encoding and computational workload. The resource manager, therefore, can take advantage of this trade-off and adjust the quality of the encoding according to the available resources, such as energy budget and number of PEs, at run-time.
A transition from one mode to any other one is possible, thereby resulting in totally 12 possible transitions.
At run-time, reconfiguration of the parameter \SYM{dl} is triggered by the environment, e.g., the resource manager in this case.
Subsequently, control actor $A_c$ propagates \SYM{dl} to the data-flow actors shown in \FIG~\ref{fig:vocoder_sacsdf} through the dashed-lined edges.

We measured the WCETs of all dataflow actors in \FIG~\ref{fig:vocoder_sacsdf} in the four modes on an ARM Cortex-A9~\cite{arm-a9} processor.
All dataflow actors were compiled using the compiler \code{arm-xilinx-eabi-gcc 4.7.2} with the vectorization option.
The WCETs of all actors in all four modes are given in \TAB~\ref{tab:vocoder_wcet}.
It is worth to note that in mode $\textit{SI}^8$, actors \SYM{Spec2Env} and \SYM{male2female} exhibit exceptionally high WCETs.
It is because parameter \SYM{dl} represents the size of the inner-most loop in the computation of actors \SYM{Spec2Env} and \SYM{male2female}.
Small \SYM{dl} (in this case $\SYM{dl}=8$) leads to the fact that the inner-most loop cannot be vectorized by the compiler.
In the other modes from $\textit{SI}^{16}$ to $\textit{SI}^{64}$, larger sizes of the inner-most loop (\SYM{dl} equal to 16, 32, and 64, respectively) lead to full vectorization of the computation of actors \SYM{Spec2Env} and \SYM{male2female}.
Therefore, in these three modes, the WCETs of actors \SYM{Spec2Env} and \SYM{male2female} are even smaller than the ones in mode $\textit{SI}^8$.
The dataflow actors of Vocoder are allocated to 4 PEs as shown in \FIG~\ref{fig:vocoder_alloc}.
This allocation guarantees that the shortest periods~(maximum throughput) in the steady-states of all modes can be achieved.

\begin{table}[!t!b]
	\centering
	\caption{WCETs of all actors in Vocoder (in clk. cycles).}
	\resizebox{0.475\textwidth}{!}{		\begin{tabular}{ |c | c | c | c | c | c | c | c | c | c | c | c | c |}
			\hline
			\begin{turn}{90} \scriptsize{Mode}\end{turn}        & 	\begin{turn}{90}\scriptsize{\SYM{ReadWave}}\end{turn}      & \begin{turn}{90}\scriptsize{\SYM{DFT}}\end{turn}    & \begin{turn}{90}\scriptsize{\SYM{AddCosWin}}\end{turn}  & \begin{turn}{90}\scriptsize{\SYM{Rec2Polar}}\end{turn} & \begin{turn}{90}\scriptsize{\SYM{Unwrap}}\end{turn}  & \begin{turn}{90}\scriptsize{\SYM{Spec2Env}}\end{turn}  & \begin{turn}{90}\scriptsize{\SYM{male2female}}\end{turn}  & \begin{turn}{90}\scriptsize{\SYM{Polar2Rec}}\end{turn}  & \begin{turn}{90}\scriptsize{\SYM{InvDFT}}\end{turn}  & \begin{turn}{90}\scriptsize{\SYM{WriteWave}}\end{turn}  \\ \hline \hline
			$\textit{SI}^8$         & 3704          & 16775   & 16  & 90  & 359  & 7168  & 1093  &3    & 236       &3660      \\\hline
			$\textit{SI}^{16}$      & 3704          & 35121   & 35  & 183 & 691  & 1163  & 138    &260  & 644       &3660       \\\hline
			$\textit{SI}^{32}$      & 3704          & 71337   & 75  & 366 & 1393  &1392    &210     &507  & 988     &3660      \\\hline
			$\textit{SI}^{64}$      & 3704          & 144531  & 150 &1156 &2346  &1696    &426     &1056 & 3630 &3660       \\\hline
	\end{tabular}}
	\label{tab:vocoder_wcet}
\end{table}
\begin{table}[!t!b]
	\centering
	\caption{Performance results of four modes of Vocoder in the steady-state.} 
		\resizebox{0.475\textwidth}{!}{			\begin{tabular}{ |c | c | c | c | c | c | c | c |}
				\hline
				Mode        & Period ($T$ in clk.)      & Total utilization ($U$)    & Iteration latency ($L$)  \\ \hline \hline
				$\textit{SI}^8$         & 917504       & 1.24                 & 7339608                    \\ \hline
				$\textit{SI}^{16}$      & 148864       & 2.36                 & 1191436                     \\\hline
				$\textit{SI}^{32}$      & 178176       & 3.19                 & 1425448                       \\ \hline
				$\textit{SI}^{64}$      & 300288       & 3.4                  & 2402550                     \\\hline
			\end{tabular}}
	\label{tab:vocoder_period}
	\end{table}

\TAB~\ref{tab:vocoder_period} shows the performance results for the four modes in their steady-state under SPS. For instance, the second column at the first row in \TAB~\ref{tab:vocoder_period} indicates that it is guaranteed for sink actor \SYM{WriteWave} to produce 256 samples per $917451$ clock cycles in mode $\textit{SI}^{8}$.
This is the ``worst-case'' performance among all four modes because the \SYM{Spec2Env} actor exhibits exceptionally high workload (\textit{cf}. WCETs in \TAB~\ref{tab:vocoder_wcet}) in mode $\textit{SI}^8$.
Consequently, actor \SYM{Spec2Env} becomes the ``bottleneck'' actor, so that mode $\textit{SI}^8$ cannot be scheduled with higher throughput~(shorter period).
Nevertheless, all mode $\textit{SI}^8$ actors as a whole require a total processor utilization ($U$) of only $1.24$ (see the third column in \TAB~\ref{tab:vocoder_period}) which is the least among all modes.
From \TAB~\ref{tab:vocoder_period}, we can see that
\SADF{} together with the SPS framework brings another advantage of efficiently utilizing PE resources.
For example, in case that Vocoder is switched to a mode with lower processor utilization, idle capacity of PEs can be efficiently utilized by admitting other applications at run-time without introducing interference to the currently running Vocoder.

Now, we focus on the performance results of the MOO protocol, namely transition delays, for all possible transitions between the four modes of Vocoder.
\begin{table}[!t!b]
	\centering
	\caption{Performance results for all mode transitions of Vocoder.} 
		 \resizebox{0.47\textwidth}{!}{		\begin{tabular}{ |c | c | c | c | c | c | c | c |}
			\hline
			Transition                                          & $\Delta_\text{min}^{o \rightarrow l}$ (in clk.) & $\Delta_\text{max}^{o \rightarrow l}$ (in clk.) & $x$ (in clk.) & $\delta^{o \rightarrow l}$ (in clk.)\\
			($\textit{SI}^o$ to $\textit{SI}^l$)                 &                      &     &                 &        \\\hline\hline
			$\textit{SI}^8 \rightarrow \textit{SI}^{64}$         & 3636815   & 4554266        & 1234264   & 1234264             \\ \hline
			$\textit{SI}^8 \rightarrow \textit{SI}^{32}$         & 2903988   & 3821439        & 1478540   & 1478540                 \\ \hline
			$\textit{SI}^8 \rightarrow \textit{SI}^{16}$         & 2728479   & 3645930        & \textbf{1537043}   & 1537043             \\ \hline
			$\textit{SI}^{16} \rightarrow \textit{SI}^{64}$      & 2402550   & 2551480         & 0           & 0                \\\hline
			$\textit{SI}^{16} \rightarrow \textit{SI}^{32}$      & 1425448   & 1574378         & 0           & 0                  \\\hline
			$\textit{SI}^{16} \rightarrow \textit{SI}^{8}$       & 7339608   & 7488538         & 0           & 0                \\\hline
			$\textit{SI}^{32} \rightarrow \textit{SI}^{64}$      & 2402550   & 2580731         & 0           & 0                \\ \hline
			$\textit{SI}^{32} \rightarrow \textit{SI}^{16}$      & 1425448   & 1603629         & 234012     & 234012            \\ \hline
			$\textit{SI}^{32} \rightarrow \textit{SI}^{8}$       & 7339608   & 7517789         & 0           & 0                \\ \hline
			$\textit{SI}^{64} \rightarrow \textit{SI}^{32}$      & 2402550   & 2702869         & 977102     & 977102             \\\hline
			$\textit{SI}^{64} \rightarrow \textit{SI}^{16}$      & 2402550   & 2702869         & 1211114   & 1211114            \\\hline
			$\textit{SI}^{64} \rightarrow \textit{SI}^{8}$       & 7339608   & 7639927         & 0           & 0                \\\hline
		\end{tabular}}
\label{tab:vocoder_trans}
\end{table}
\TAB~\ref{tab:vocoder_trans} shows both the minimum and maximum transition delays in accordance with Theorem \ref{theo:min_max_delay} for all transitions.
We can see in the second column of \TAB~\ref{tab:vocoder_trans} that, in the best case, the transition delays for 6 out of 12 transitions remain the same as the iteration latencies of the new modes.
This can be seen as $x = 0$ shown in the fourth column.
In these 6 transitions, the proposed MOO protocol does not introduce any extra delay.
In the 6 remaining transitions, as expected, the MOO protocol introduces offset $x > 0$ to the transitions from an old mode with a longer iteration latency to a new mode with a shorter iteration latency.
For instance, the largest $x$ (in bold shown in \TAB~\ref{tab:vocoder_trans}) happens in case of a transition from mode $\textit{SI}^8$ with the longest iteration latency (see the fourth column in \TAB~\ref{tab:vocoder_period}) to mode $\textit{SI}^{16}$ with the shortest iteration latency.
To quantify $x$, we compute the percentage of $x$ compared to both minimum and maximum transition delays as
\begin{equation*}
	\Omega_\text{min} = \frac{x}{\Delta_\text{min}^{o\rightarrow l}}\times 100\%,~~\Omega_\text{max} = \frac{x}{\Delta_\text{max}^{o\rightarrow l}}\times 100\%.
\end{equation*}
$\Omega_\text{min}$ varies from the worst-case $56\%$ to the best case $16\%$ with an average of $41\%$, whereas $\Omega_\text{max}$ varies from the worst-case $44\%$ to the best case $14\%$ with an average of $33\%$.
Therefore, the increase of the transition delays due to the MOO protocol is reasonable for this real-life application.

Next, we consider the effect of the actor allocation shown in \FIG~\ref{fig:vocoder_alloc} on the earliest starting times of actors in the new mode upon a transition~(\textit{cf.} Lemma \ref{lm:overload_offset}).
In this particular example, we find out that no extra delay is incurred to any actor in all transitions due to the fixed actor allocation.
This can be seen from the fourth and fifth columns in \TAB~\ref{tab:vocoder_trans}, where $\delta^{o\rightarrow l} = x$.

\begin{table}[!t!b]
	\centering
	\caption{The period and iteration latency of modes in MP3 decoder in clk. cycles.} 
	\resizebox{0.475\textwidth}{!}{		\begin{tabular}{ |c | c | c | c | c | c | c | c |}
			\hline
			Mode	&	s-s	&	s-l	&	l-s	&	l-l	&	m \\ \hline \hline
			Period ($T$)	& 5830000	&	5785970	&	 5830000	&	4640000	&	5760000                    \\ \hline
			Iteration latency ($L$)       & 9434720   & 9234570  & 9278600 & 7466400 & 9089900                     \\\hline
	\end{tabular}}
	\label{tab:period-latency-mp3}
\end{table}
\begin{table*}[!t!b]
	\centering
	\caption{Performance results of MP3 decoder for four different mode transition sequences using MADF and FSM-SADF models.} 
	\resizebox{1\textwidth}{!}{		\begin{tabular}{ |c|c|c|c|c|c|c|c||c|c|c|c|c|c|c|}
			\cline{1-15}
			
			\multirow{2}{*}{Mode Sequence} & \multicolumn{7}{ |c || }{FSM-SADF\cite{Geilen2010_codes10}} & \multicolumn{7}{ | c| }{MADF} \\ \cline{2-15}
			& \multicolumn{4}{ |c | }{Iteration latency} & \multicolumn{3}{ |c || }{Transition delay} & \multicolumn{4}{ |c | }{Iteration latency} & \multicolumn{3}{ |c | }{Transition delay}\\ \hline \hline
			
			\multirow{2}{*}{s-s$\rightarrow$s-l$\rightarrow$m$\rightarrow$l-l} &$L^{\text{s-s}} $&$L^{\text{s-l}}$& $L^{\text{m}}$& $L^{\text{l-l}}$ &  $\Delta^{\text{s-s} \rightarrow \text{s-l}}$& $\Delta^{\text{s-l}\rightarrow \text{m}}$ & $\Delta^{\text{m} \rightarrow \text{l-l}}$&$L^{\text{s-s}} $&$L^{\text{s-l}}$& $L^{\text{m}}$& $L^{\text{l-l}}$ &  $\Delta^{\text{s-s} \rightarrow \text{s-l}}$& $\Delta^{\text{s-l}\rightarrow \text{m}}$ & $\Delta^{\text{m} \rightarrow \text{l-l}}$ \\   \cline{2-15} & 9434720 & 9434670 & 9310400 & \textbf{9310400} & 9434670 &  \textbf{9310400} & 9310400 & 9434720 & 9234570& 9089900 &  \textbf{7466400} & 10032600 &  \textbf{9261700} & 9089900\\ \hline \hline
			
			\multirow{2}{*}{s-s$\rightarrow$l-l$\rightarrow$s-l$\rightarrow$m} &$L^{\text{s-s}} $&$L^{\text{l-l}}$& $L^{\text{s-l}}$& $L^{\text{m}}$ &  $\Delta^{\text{s-s}\rightarrow \text{l-l}}$& $\Delta^{\text{l-l}\rightarrow \text{s-l}}$ & $\Delta^{\text{s-l}\rightarrow \text{m}}$&$L^{\text{s-s}} $&$L^{\text{l-l}}$& $L^{\text{s-l}}$& $L^{\text{m}}$ &  $\Delta^{\text{s-s}\rightarrow \text{l-l}}$& $\Delta^{\text{l-l}\rightarrow \text{s-l}}$ & $\Delta^{\text{s-l}\rightarrow \text{m}}$ \\   \cline{2-15} & 9434720 & \textbf{9434700}  & 9434670 & 9217800 & 9434700 & 9434670 &  \textbf{9217800} & 9434720& \textbf{7466400} & 9234570 & 9089900 & 9434700 & 9234500 &  \textbf{9261700}\\ \hline \hline
			
			\multirow{2}{*}{l-s$\rightarrow$s-l$\rightarrow$m$\rightarrow$l-l} &$L^{\text{l-s}} $&$L^{\text{s-l}}$& $L^{\text{m}}$& $L^{\text{l-l}}$ & $\Delta^{\text{l-s}\rightarrow \text{s-l}}$& $\Delta^{\text{s-l}\rightarrow \text{m}}$ & $\Delta^{\text{m}\rightarrow \text{l-l}}$&$L^{\text{l-s}} $&$L^{\text{s-l}}$& $L^{\text{m}}$& $L^{\text{l-l}}$ &  $\Delta^{\text{l-s}\rightarrow \text{s-l}}$& $\Delta^{\text{s-l}\rightarrow \text{m}}$ & $\Delta^{\text{m}\rightarrow \text{l-l}}$\\   \cline{2-15} & 9278600 & 9278570 & 9197200 & \textbf{9197200} & 9278570 &  \textbf{9197200} & 9197200 & 9278600 & 9234570 & 9089900 & \textbf{7466400} & 9876500 & \textbf{9261700}& 9089900 \\ \hline \hline			 
			
			\multirow{2}{*}{s-s$\rightarrow$l-s$\rightarrow$s-l$\rightarrow$l-l} &$L^{\text{s-s}}$ & $L^{\text{l-s}}$ & $L^{\text{s-l}}$& $L^{\text{l-l}}$ &  $\Delta^{\text{s-s}\rightarrow \text{l-s}}$& $\Delta^{\text{l-s}\rightarrow \text{s-l}}$ & $\Delta^{\text{s-l}\rightarrow \text{l-l}}$&$L^{\text{s-s}}$ & $L^{\text{l-s}}$ & $L^{\text{s-l}}$& $L^{\text{l-l}}$ &  $\Delta^{\text{s-s}\rightarrow \text{l-s}}$& $\Delta^{\text{l-s}\rightarrow \text{s-l}}$ & $\Delta^{\text{s-l}\rightarrow \text{l-l}}$\\   \cline{2-15} & 9434720 & 9434700 & 9434670 & \textbf{8661500} & 9434700 & 9434670 & 8661500 & 9434720& 9278600 & 9234570 & \textbf{7466400} & 10032600 & 9876500 & 9234600 \\ \hline 			 
	\end{tabular}}
	\label{tab:mp3-performance}
\end{table*}
\subsection{Case Study 2}
	To further evaluate the MOO protocol, presented in \SEC~\ref{sec:moo_trans}, in this section, we performed an experiment with the MP3 decoder application, which is a real-life adaptive streaming application, taken from \cite{Geilen2010_codes10}. This MP3 decoder is a frame-based algorithm that retrieves audio frames from the incoming compressed bitstream. In the MP3 decoder, each audio frame can be decoded using a different method. In total, MP3 decoder has five individual decoding methods for audio frames that are denoted as \{s-s, l-l, l-s, s-l, m\}.

Each of these methods can be represented accurately by an SDF graph. Therefore, the application behavior can be accurately captured using FSM-SADF \cite{Geilen2010_codes10} rather than conservatively capture these methods in a static dataflow model. Consequently, a much tighter performance can be guaranteed by FSM-SADF graph than SDF. Note that since each mode in our MADF model is represented as a CSDF graph, our MADF is more expressive than FSM-SADF and therefore, the MP3 decoder can be also properly modeled with MADF. The period and iteration latency of each mode are given in \TAB~\ref{tab:period-latency-mp3}.

Let us now compare the throughput of MP3 decoder modeled as MADF and FSM-SADF graphs. To compute the throughput of MP3 decoder modeled by the FSM-SADF, we use the publicly available SDF$^3$ tool set \cite{stuijk2006sdf}. Since the type of frames may change non-deterministically in arbitrary orders, SDF$^3$ detects the worst-case mode transition using the state-space exploration approach developed in \cite{Geilen2010_codes10} for FSM-SADF to lower bound the throughput. To compute the worst-case throughput of the application, we use the \texttt{sdf3analysis-fsmsadf} tool from SDF$^3$. Similarly, we use the same approach to compute the throughput of our MADF model that uses the MOO protocol. For both models, the same throughput of 1.75$\cdot10^{(-7)}$ frame per clock cycle is achieved. Therefore, both models perform equally well in terms of the worst-case throughput they can guarantee and the delay introduced by our MOO protocol during mode transitions has no impact on the worst-case throughput.

Now, we focus on the performance results of our MADF and FSM-SADF models in terms of the iteration latency of the modes and the transition delay. The results of this comparison for four different mode transition sequences is give in \TAB~\ref{tab:mp3-performance}. In this table, for each mode transition sequence, the iteration latency of each mode and the transition delay of each mode transition are given for our MADF model that uses the MOO protocol and the FSM-SADF model that uses the self-timed protocol. From this table, we can clearly see that our MADF retains the iteration latency of each mode irrespective of the mode transition sequences. Using the FSM-SADF model, however, the iteration latency of modes in the steady-state is accordingly changed with respect to the order of mode transitions. For instance, mode l-l has different iteration latency, $L^{\text{l-l}}$, of 9310400, 9434700, 9197200, and 8661500 for the different mode transition sequences, when using FSM-SADF. In contrast, the same mode l-l has a constant iteration latency of 7466400 under our MADF model (bolded in \TAB~\ref{tab:mp3-performance}). Therefore, the iteration latency of modes in the steady-state can not be guaranteed under the FSM-SADF model as it is highly dependent on the order of mode transitions which is not known beforehand at design-time.

From \TAB~\ref{tab:mp3-performance}, we can also see that by changing the iteration latency of the modes, the transition delays are also changed. Although the transition delays are sometimes shorter in the FSM-SADF model, the FSM-SADF model is potentially unpredictable. Our MADF model, however, is completely predictable because the (minimum) transition delays for all mode transitions can be computed beforehand at design-time according to Theorem \ref{theo:min_max_delay}. For instance, the transition from mode s-l to mode m has different transition delay, $\Delta^{\text{s-l}\rightarrow \text{m}}$, of 9310400, 9217800, and 9197200 for different mode transition sequences under the FSM-SADF model whereas this mode transition has a constant transition delay of 9261700 under our MADF model (bolded in \TAB~\ref{tab:mp3-performance}).
 \section{Conclusion}\label{sec:Conclusion}
In this paper, we have proposed the novel Mode-Aware Data Flow (MADF) model which can capture effectively the adaptive nature of modern streaming applications. Moreover, as an important part of the operational semantics of MADF, we have proposed a novel protocol for mode transitions. The main advantage of this transition protocol is that, in contrast to the self-timed transition protocol, it avoids timing interference between modes upon mode transitions. As a result, any mode transition can be analyzed independently from others that occurred in the past. Furthermore, based on the transition protocol, we have proposed a hard real-time analysis and scheduling framework to reason and guarantee timing constraints by avoiding processor overloading during mode transitions. Finally, we evaluate the effectiveness of our MADF model compared with the well-know FSM-SADF model by conducting two case studies using two real-life adaptive streaming applications. 
\bibliography{references}

\begin{thebibliography}{10}

\bibitem{Gerstlauer_2009}
A.~Gerstlauer~et al.
\newblock Electronic system-level synthesis methodologies.
\newblock {\em IEEE Trans. Comput.-Aided Design Integr. Circuits Syst.}, 2009.

\bibitem{Lee_1987}
E.~A. Lee and D.~G. Messerschmitt.
\newblock Static scheduling of synchronous data flow programs for digital
  signal processing.
\newblock {\em IEEE Trans. on computers}, 1987.

\bibitem{Bilsen1996}
G.~Bilsen~et al.
\newblock Cycle-static dataflow.
\newblock {\em IEEE Trans. Signal Process.}, 1996.

\bibitem{Theelen2006_memocode}
B.~D. Theelen~et al.
\newblock A scenario-aware data flow model for combined long-run average and
  worst-case performance analysis.
\newblock In {\em MEMOCODE}, 2006.

\bibitem{Geilen2010_codes10}
M.~Geilen and S.~Stuijk.
\newblock Worst-case performance analysis of synchronous dataflow scenarios.
\newblock In {\em CODES+ISSS}, 2010.

\bibitem{Wiggers2011_tecs}
M.~H. Wiggers~et al.
\newblock Buffer capacity computation for throughput-constrained modal task
  graphs.
\newblock {\em ACM Trans. Embed. Comput. Syst.}, 2010.

\bibitem{Moreira2012}
O.~Moreira.
\newblock Temporal analysis and scheduling of hard real-time radios running on
  a multi-processor.
\newblock {\em ser. PHD Thesis, Technische Universiteit Eindhoven}, 2012.

\bibitem{Bhattacharya2001}
B.~Bhattacharya and S.~S. Bhattacharyya.
\newblock Parameterized dataflow modeling for dsp systems.
\newblock {\em IEEE Trans. Signal Process.}, 2001.

\bibitem{Geilen2011_tecs}
M.~Geilen.
\newblock Synchronous dataflow scenarios.
\newblock {\em ACM Trans. Embed. Comput. Syst.}, 2010.

\bibitem{henia2007scenario}
R.~Henia and R.~Ernst.
\newblock Scenario aware analysis for complex event models and distributed
  systems.
\newblock In {\em RTSS}, 2007.

\bibitem{Negrean2011_etfa}
M.~Negrean~et al.
\newblock Bounding mode change transition latencies for multi-mode real-time
  distributed applications.
\newblock In {\em ETFA}, 2011.

\bibitem{Real2004}
J.~Real and A.s Crespo.
\newblock Mode change protocols for real-time systems: A survey and a new
  proposal.
\newblock {\em Real-time Systems}, 2004.

\bibitem{Stoimenov2009_date}
N.~Stoimenov~et al.
\newblock Reliable mode changes in real-time systems with fixed priority or edf
  scheduling.
\newblock In {\em DATE}, 2009.

\bibitem{Bamakhrama2011_emsoft}
M.~Bamakhrama and T.~Stefanov.
\newblock On the hard-real-time scheduling of embedded streaming applications.
\newblock {\em DAES}, 2013.

\bibitem{Lickly2008_pret}
B.~Lickly~et al.
\newblock Predictable programming on a precision timed architecture.
\newblock In {\em CASES}, 2008.

\bibitem{Liu1973_hrt}
C.~L. Liu and J.~W. Layland.
\newblock Scheduling algorithms for multiprogramming in a hard-real-time
  environment.
\newblock {\em Journal of the ACM (JACM)}, 1973.

\bibitem{Stuijk2008_tc}
S.~Stuijk~et al.
\newblock Throughput-buffering trade-off exploration for cyclo-static and
  synchronous dataflow graphs.
\newblock {\em IEEE Trans. on Computers}, 2008.

\bibitem{Bodin2013_estimedia}
B.~Bodin~et al.
\newblock Periodic schedules for cyclo-static dataflow.
\newblock In {\em ESTIMedia}, 2013.

\bibitem{Neuendorffer2004}
S.~Neuendorffer and E.~Lee.
\newblock Hierarchical reconfiguration of dataflow models.
\newblock In {\em MEMOCODE}, 2004.

\bibitem{niknam2017energy}
S.~Niknam and T.~Stefanov.
\newblock Energy-efficient scheduling of throughput-constrained streaming
  applications by periodic mode switching.
\newblock In {\em SAMOS}, 2017.

\bibitem{Gordon2006_asplos}
M.~I. Gordon, W.~Thies, and S.~Amarasinghe.
\newblock Exploiting coarse-grained task, data, and pipeline parallelism in
  stream programs.
\newblock {\em ACM SIGOPS Operating Systems Review}, 2006.

\bibitem{arm-a9}
ARM Cortex.
\newblock A9 processor.
\newblock {\em
  http://www.arm.com/products/processors/cortex-a/cortex-a9.php,~2013}.

\bibitem{stuijk2006sdf}
S.~Stuijk, M.~Geilen, and T.~Basten.
\newblock Sdf\^{} 3: Sdf for free.
\newblock In {\em ACSD}, 2006.

\end{thebibliography}
\bibliographystyle{unsrt}

\begin{IEEEbiography}[{\includegraphics[width=1in,height=1.25in,clip,keepaspectratio]{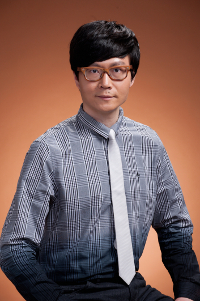}}]
{Jiali Teddy Zhai} was born on 16th of October, 1982. In September 2009, he received Diplom Informatik (Master Degree in Computer Science) from Friedrich-Alexander Universitat Erlangen-Nurnberg, Germany. During his study, Teddy worked at Institute for Hardware-Software-Co-Design headed by Prof. Jürgen Teich with the focus on designing high-level synthesis tools targeting high-performance computing systems based on FPGA platforms. In October 2009, Teddy joined the Leiden Embedded Research Center (LERC) which is part of the Leiden Institute of Advanced Computer Science (LIACS) at Leiden University. He was appointed as a research and teaching assistant (Ph.D. student). He was involved in the NEtherlands STreaming (NEST) project in collaboration with NXP semiconductor, Philips Healthcare, etc. In April 2014 Teddy joined Irdeto B.V. Netherlands as Senior Security Engineer. From June 2016 on, Teddy has been appointed by Green Hills Software Netherlands as Security Solution Architect, Europe.
\end{IEEEbiography}

\begin{IEEEbiography}[{\includegraphics[width=1in,height=1.25in,clip,keepaspectratio]{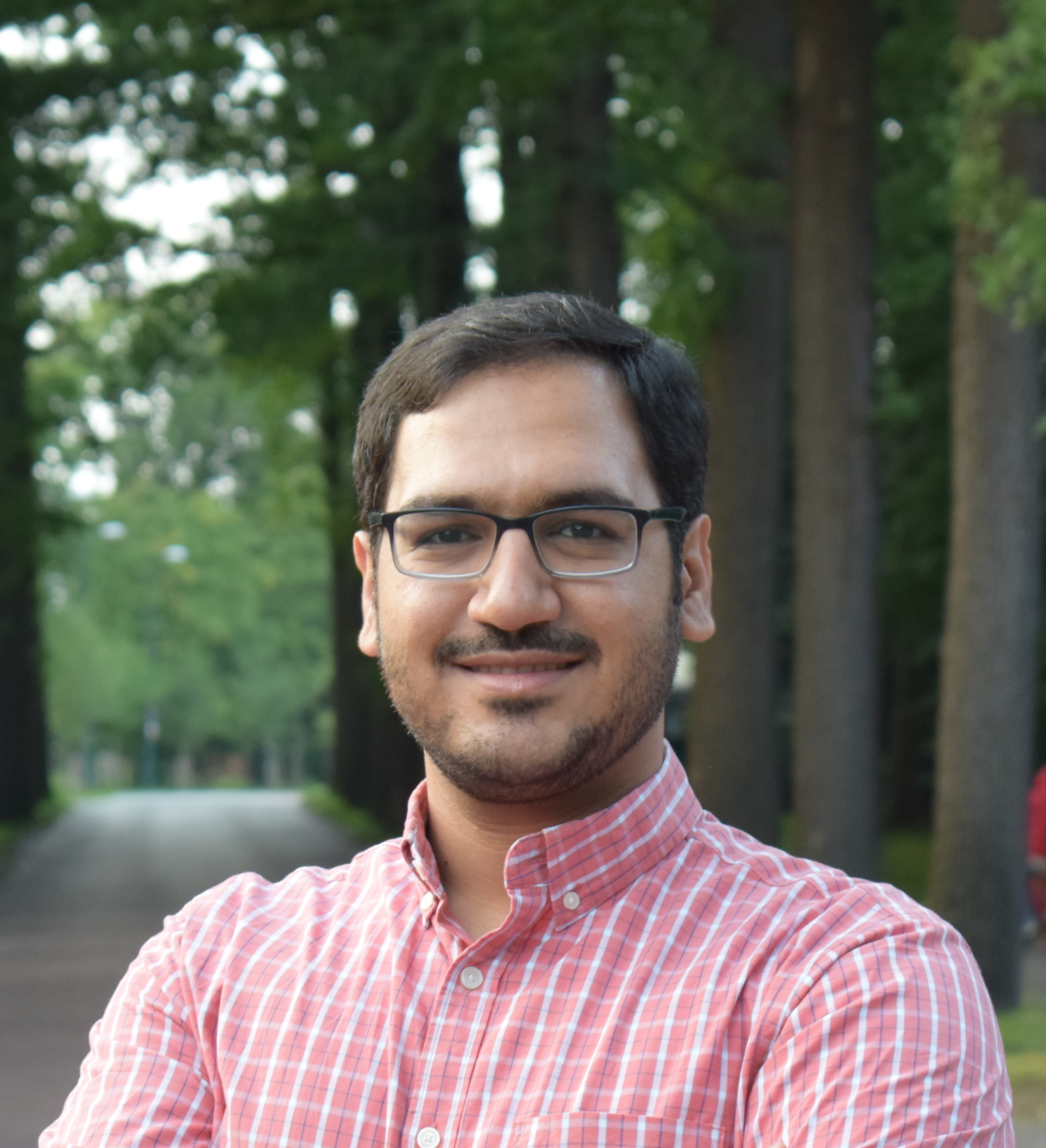}}]
	{Sobhan Niknam} received his B.Sc. and M.Sc. degrees in computer engineering from Shahed University and Iran University of Science and Technology both in Tehran, Iran, in 2012 and 2014, respectively. Since March 2015, he has been with the Leiden Embedded Research Center (LERC) which is part of the Leiden Institute of Advanced Computer Science (LIACS) at Leiden University in The Netherlands, where he is working toward the Ph.D. degree in computer science. His research interests include real-time embedded systems and system-level multi-core systems design.
\end{IEEEbiography}

\begin{IEEEbiography}[{\includegraphics[width=1in,height=1.25in,clip,keepaspectratio]{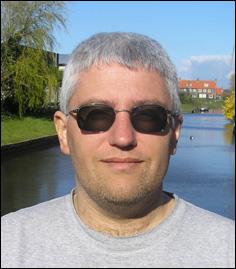}}] 
	{Todor Stefanov} (S'01--M'05) received the Dipl.Ing. and M.S. degrees in computer engineering from The Technical University of Sofia, Bulgaria, in 1998 and the Ph.D. degree in computer science from Leiden University, The Netherlands, in 2004. Currently, he is an associate professor in the Leiden Institute of Advanced Computer Science at Leiden University and the head of the Leiden Embedded Research Center (LERC) which is a medium-size
	research group with a strong track record in the area of system-level modeling and synthesis, programming, and implementation of heterogeneous embedded systems. Dr. Stefanov is a recipient of the prestigious 2009 IEEE TCAD DONALD O.PEDERSON BEST PAPER AWARD for his journal article "Systematic and Automated Multi-processor System Design, Programming, and Implementation" published in the IEEE Transactions on Computer-Aided Design of Integrated Circuits and Systems (TCAD). He is editorial board member of the Springer Journal on Embedded Systems. He has also been editorial board member of the International Journal of Reconfigurable Computing and guest associate editor of ACM Transactions on Embedded Computing Systems (2013). He has been General Chair of ESTIMedia 2015 and Local Organization Co-Chair of ESWeek 2015. Moreover, he serves (has served) on the organizational committees of several leading conferences, symposia, and workshops, such as DATE, ACM/IEEE CODES+ISSS, RTSS, IEEE ICCD, IEEE/IFIP VLSI-SoC, ESTIMedia, SAMOS (as TPC member), and IEEE ESTIMedia, ACM SCOPES (as Program Chair). Dr. Stefanov (co-)authored more than 80 scientific papers. His research interests include several aspects of embedded systems design, with particular emphasis on system-level design automation, multiprocessor systems-on-chip design, and hardware/software co-design.
\end{IEEEbiography}

\end{document}